%% file: equilibrium_commodity_170121.tex
\documentclass[a4paper,DIV=classic]{myart}
\usepackage{todonotes}
\usepackage{xspace}
\usepackage{stackrel}
\usepackage{tikz}
\usetikzlibrary{matrix}

\input{head}
\input{ahead}

\DeclareMathAlphabet{\mathpzc}{OT1}{pzc}{m}{it}

\newcommand{\UUU}{\mathcal{U}}
\newcommand{\QQQ}{\mathrm{q}}%{\mathnormal{q}}
\newcommand{\sss}{\mathpzc{s}}

\begin{document}

\title{An equilibrium model for spot and forward prices of 
commodities\tnoteref{t}}
\tnotetext[t]{We are grateful to Ulrich Horst, Kasper Larsen, Dilip Madan, 
Christoph Mainberger, Steven Shreve, Ronnie Sircar and Stathis Tompaidis for 
helpful discussions and suggestions. We thank two anonymous referees for their 
valuable comments that have significantly improved the paper. We also thank 
seminar participants at Carnegie Mellon University, the 11th German Probability 
and Statistics Days in Ulm, the Workshop in New Directions in Financial 
Mathematics and Mathematical Economics in Banff, the Workshop on Stochastic 
Methods in Finance and Physics in Crete, the International Conference in 
Advanced Finance and Stochastics in Moscow, the Conference on Frontiers in 
Financial Mathematics in Dublin, the Conference on Stochastics of Environmental 
and Financial Economics in Oslo, the 8th World Congress of the Bachelier Finance 
Society in Brussels, and the Seminar in Stochastic Analysis and Stochastic 
Finance in Berlin for their comments. Financial support from the IKYDA project 
``Stochastic Analysis in Finance and Physics'' is gratefully acknowledged.}

\author[a,1]{Michail Anthropelos}
\author[b,2]{Michael Kupper}
\author[c,3]{Antonis Papapantoleon}

\address[a]{University of Piraeus, 80 Karaoli and Dimitriou Str., 18534
        Piraeus, Greece}
\address[b]{University of Konstanz, Universit\"atstra\ss e 10, 78464 Konstanz}
\address[c]{Institute of Mathematics, TU Berlin, Stra\ss e des 17. Juni 136,
        10623 Berlin, Germany}

\eMail[1]{anthropel@unipi.gr}
\eMail[2]{kupper@uni-konstanz.de}
\eMail[3]{papapan@math.tu-berlin.de}

%\myThanks[s]{}

\date{}

\abstract{We consider a market model that consists of financial investors and 
producers of a commodity. Producers optionally store some production for future 
sale and go short on forward contracts to hedge the uncertainty of the future 
commodity price. Financial investors take positions in these contracts in order 
to diversify their portfolios. The spot and forward equilibrium commodity prices 
are endogenously derived as the outcome of the interaction between producers 
and investors. Assuming that both are utility maximizers, we first prove the 
existence of an equilibrium in an abstract setting. Then, in a framework where 
the consumers' demand and the exogenously priced financial market are 
correlated, we provide semi-explicit expressions for the equilibrium prices and
analyze their dependence on the model parameters. The model can explain why 
increased investors' participation in forward commodity markets and higher 
correlation between the commodity and the stock market could result in higher 
spot prices and lower forward premia.}

\keyWords{Commodities, equilibrium, spot and forward prices, forward premium,
stock and commodity market correlation.}

\keyJELClassification{Q02, G13, G11, C62}
\keyAMSClassification{91B50, 90B05}

\maketitle\frenchspacing

\section{Introduction}\label{sec:intro}

Since the early 2000s, the futures and forward contracts written on commodities 
have been a widely popular investment asset class for many financial 
institutions. As indicatively reported in \cite{CFTC08}, the value of 
index-related futures' holdings in commodities grew from \$15 billion in 2003
to more than \$200 billion in 2008.\footnote{According to a recent estimation
by Barclays Capital, the total commodity-linked assets were around \$325
billion at the end of June 2014 (see Barclays Investment Bank 2014 and
\cite{HendPeaWan14}).} This significant inflow of funds has coincided, up to
2008, with a steep increase in the spot and futures prices of the majority of
commodities, especially the ones included in popular commodity indices. The
comovement of amounts invested in commodity-linked securities and the prices of
the associated commodities continued even during the prices' bust in 2008 and
their recovery which started in 2009; see e.g. the empirical studies presented
in \citet{Singleton14}, \citet{TangXio2012} and \citet{BuyRob14}.\footnote{In
\cite{Singleton14}, it is shown that investors' index positions in crude oil
are highly correlated with crude oil prices, while \cite{BuyRob14} provides
statistical evidence which indicates that the excess speculation in U.S.
Commodity Futures Markets increased from 11\% in 2000 to more than 40\% in
2008. In theoretical terms, this correlation affects heavily the agents'
optimization problems (see among others \cite{CadHau06}).} Furthermore, several
statistical studies have found that the correlation between commodity prices
and the stock market has grown during the last years. As an example,
\cite{BuyRob14} argues that the correlation of the U.S. stock market (weekly)
returns and the returns of the GSCI commodity indices varies from -38\% to 40\%
depending on the period, and stays positive and away from zero after 2009;
further statistical evidence on the increased correlation are given in
\cite{TangXio2012}, \cite{Singleton14} and in \citet{SilThorp13}. Therefore,
the investment strategies of financial institutions on stock and commodity 
markets should be considered in the same optimization problem and not 
independently.

The booms and busts of the prices of major commodities during the last decade
has naturally captured the interest of the academic community. The main question
addressed is whether the behavior of commodities' prices is caused by the
(enhanced by the financialization) positions of speculators or by the
fluctuations of fundamental economic factors (i.e. increased demand and weakend
supply).\footnote{The opinion that speculative forces are the main reason for
the booms and busts of commodity prices (especially in oil and gas markets) is
supported by the empirical studies in
\cite{HendPeaWan14,Singleton14,TangXio2012}, while
\cite{BuyRob14,HamWu12,RouTang12,StollWha2010} provide statistical tests that
are in favor of the fundamental economic reasoning (see also
\cite{JuvPetre12,KilMur14}). For a more detailed literature survey on this
debate, we refer the reader to \cite{HendPeaWan14} and \cite{RouTang12}.} Even
though there exist several empirical studies, the theoretical approaches that
link spot and forward prices of commodities with the rest of the investment
assets are scarce.

The main goal of this paper, is to establish an equilibrium model that allows
to endogenously derive both the spot and the forward price of a commodity, and
is flexible and general enough to include not only the randomness of the
commodity demand and the commodity holders' storage option, but also risk averse
agents and correlation between the stock and the commodity market. In our model,
equilibrium commodity prices are formed as the outcome of the interaction
between market participants, and simultaneously clear out the spot and the
forward market. The forces that lead to the market equilibrium are the
producers' goal to maximize their spot revenues and optimally hedge the risk of
the future commodity price, and the investors' goal to achieve an optimal
portfolio strategy that, besides the stock market, includes also a position in
the commodity's forward contracts. The model offers new insights on how specific
model inputs, such as the agents' risk aversion, the correlation of the stock
and commodity market and the uncertainty of the future commodity price,
influence the equilibrium prices and the related risk premia.

\subsection{Model description}

We consider a model of two points in time: the initial one and a given
(short-term) future horizon $T$. We assume that the main market participants are
the representative agents of the commodity's holders/suppliers and the financial
investors/speculators, who shall hereafter be called
\textit{producers}\footnote{We refer to the commodity sellers as producers by
following the related literature (see e.g.
\cite{AchaLochRama13,Baker12,RouTang12}). Some authors impose that the
representative agent of the supply side is the commodity refiners or storage
managers who hold the production and in some cases control its supply in the
market (see for instance \cite{EkeLauVil2014,Pirrong11}). In our model, the
production schedule is a given input and the spot revenues of the producers come
only from the commodity sales, thus our findings apply directly in case the
refiners are the ones that distribute the commodity in the market.} and
\textit{investors}\footnote{As highlighted in \cite{KniPin13}, it is rather
difficult to identify whether the long position in the commodity forward
contracts is taken by investors who just want to diversify or by speculators who
invest based on specific predictions about the move of commodity prices. As a
matter of fact, the holders of a long position in the forward contract can also
be called insurers, since they undertake some of the producers' risk. Without
attempting to enter in this debate, we will call the producers' counterparties
in the forward contract investors.} respectively.

The producers' source of income are the revenues from spot and future sales.
While the commodity spot price could be determined by the spot commodity demand
function, the future price is subject to demand shocks. Assuming that producers
are risk averse, their goal is not only to maximize their spot revenues, but
also to reduce their risk exposure to the future commodity price by maximizing
their expected utility.\footnote{The massive use of derivatives by natural gas
and crude oil producers presented in Table 1 of \cite{AchaLochRama13} is a
clear evidence that commodity producers are indeed risk averse regarding their
future revenues.} If the production schedule at the initial and future time is
a predetermined pair of units, producers have two decisions to make: what
amount of the production to supply in the spot market (inventory management) and
what position to take in the forward contract (hedging strategy). Provided they
know the demand function of the commodity's consumers at the initial time, they
can determine the commodity spot price by choosing the amount of the inventory
they will hold up to the terminal time. However, random demand shocks at time
$T$ will shift the whole demand function to lower or higher levels (for
instance, in Section \ref{model-ctdm} we suppose that the random shift of the
demand function is driven by a vector of stochastic market factors). Producers 
hedge the risk which stems from the future time demand function by taking a 
short position in forward contracts written on the same commodity and with 
maturity equal to $T$. The fact that the inventory will also be sold at time $T$ 
makes the forward hedging position even more important for the 
producers.\footnote{The significance of the inventory policy in the spot and 
forward price fluctuations has been highlighted by many authors, see e.g. 
\cite{Hamilton09,KilMur14,RouSeppSpat00,Singleton14} for detailed discussions 
and statistical evidence, especially in the popular example of crude oil 
prices.}

The producers' hedging demand is covered by financial investors, who take the
opposite position in the forward commodity contracts and thus share some of the
future price uncertainty risk, possibly against a premium. They invest
optimally in an exogenously priced stock market\footnote{Here, exogenously
means that investors in the commodity forward contract are price takers in the
stock market. This implies that the volume in the commodity forward contracts
is not so large to influence the price of the stock market, an assumption that
is supported by the corresponding market volumes.} and are willing to take the
future commodity price risk in order to better diversify their portfolio.
Indeed, as mentioned above, the correlation between commodity and stock market
indices has been shown to be away from zero. This correlation could be
incorporated in a model where the stock market price is driven by the same
stochastic factors that drive the evolution of the commodity demand function.
Given this correlation, the optimal investment strategy in the stock and the
commodity market should be considered in the same optimization problem. As in
the producers' side, we assume that the investors are represented by an agent
who is a utility maximizer and whose investment choices are the (possibly
dynamic) trading strategy in the stock market and the position in the forward
commodity contract.

The optimization choices of both producers and investors clearly depend on
the forward commodity price. We define as \textit{equilibrium forward price} the
price that clears the forward market and at which both participants' expected
utilities are maximized. Given the forward price, the producers optimally choose
the inventory policy, which in turn gives the initial commodity supply and thus
determines the spot price through the initial demand function of the consumers.
Therefore, by deriving the equilibrium forward commodity price, the equilibrium
spot price as well as the producers' optimal inventory policy are also
endogenously derived in our model.

\subsection{Findings and contributions}

The main mathematical result of this work is to prove, under CARA preferences 
and upon some minor technical assumptions, that equilibrium spot and forward 
commodity prices exist. In this proof, we use standard duality arguments to show 
that both producers' and investors' optimization problems are well-defined and 
admit finite solutions. The existence of the forward commodity equilibrium price 
is first proved for every fixed producers' storage choice. Then, we show that 
this equilibrium is \textit{stable} with respect to the storage choice, which in 
turn guarantees the existence of the commodity market equilibrium. In fact, this 
stability result is interesting in its own right, since it shows that the 
market clearing is stable with respect to any control variable which belongs to 
a closed set of real numbers. 

The constructive nature of the aforementioned proof allows us to derive
implicit formulas for spot and forward prices, which can be used to investigate
how the main model parameters influence the commodity market. We illustrate this
using two examples with factors driven by \lev processes; first a Brownian
motion and then a jump-diffusion process. The main results are given below.

Focusing first on the equilibrium spot price, our results imply it is monotonic
with respect to the agents' risk aversion coefficient; it is increasing for
producers and decreasing for investors (see Figures \ref{fig:BM-1} and
\ref{fig:BM-2}). When producers are more concerned about the commodity future
price uncertainty, they increase their position in the forward contract and
hence lock the selling price at the terminal time. As long as the hedging
position is counterpartied by the investors, high risk averse producers can
increase their certain revenues today and at the same time hedge their future
price risk. What is very important in this monotonicity is the correlation
between the stock market and the demand random shocks. As it is illustrated in
Section \ref{sec:illu}, when correlation is away from zero, the equilibrium spot
price is pushed upwards. This is mainly because higher correlation means that
investors are able to better hedge their risk exposure in the commodity market
by adjusting their investments in the stock market. Hence, they are willing to
receive a lower forward premium, thus making hedging cheaper for the producers.
In particular, when the investors are more risk averse, they reduce their
share in the future commodity price risk;  thus, producers cannot hedge their
future price risk which forces them to increase their supply in the spot market.
Therefore, the spot price is a decreasing function of investors' risk aversion
(all else equal). Also, it follows that the producers' ability to hedge their 
risk tends to increase the spot equilibrium prices. Moreover, our model yields 
that the existence of a forward contract in the commodity market stabilizes 
spot prices when there is scarcity of the commodity at the terminal time; in 
particular, the presence of forward contracts increases the current spot and 
decreases the expected future spot price of the commodity.

The monotonicity of the spot price with respect to risk aversions could also be
used to explain how the participation of investors in the commodity forward
markets could result in an increase of spot commodity prices. Indeed, under CARA
preferences the more investors participate in the market, the higher the
aggregate risk tolerance becomes or, equivalently, the lower the representative
risk aversion' coefficient becomes (see, among others, \citet{Wilson68}). As
discussed above, this implies higher spot commodity price, a result that is
consistent with the observed market data (see e.g. \cite{BuyRob14} and
\citet{HendPeaWan14}). Similarly, we verify that more producers (of the same
total production) implies lower spot price. 

Besides equilibrium commodity spot prices, our model allows to endogenously
derive quantities that characterize the two major, and not mutually exclusive,
theories of forward commodity markets: the \textit{theory of storage} and the
\textit{theory of normal backwardation}. Based on the ideas introduced in
\citet{Kaldor39,Working49} and \citet{Brennan58}, the theory of storage states
that the holders of the commodity inventories get an implicit benefit, called
\textit{convenience yield}, which implies the value of the spot commodity
consumption. This yield can be approximated by the difference between the spot
and the forward price minus the cost of storage. Our equilibrium model verifies
that the convenience yield is increasing with respect to the producers' risk
aversion, meaning that the more sensitive about the risk the producers are, the
more commodity forward units they hedge depressing the forward price (see Figure
\ref{fig:BM-3} for the Brownian motion example). A similar increasing relation
holds for the investors (these relations, in particular, generalize the
results of Proposition 1 in \citet{AchaLochRama13}). However, the convenience
yield is not always monotonic with respect to the correlation coefficient. As
discussed in Section \ref{model-ctdm}, there are two effects of opposing
direction on the convenience yield, one coming from the decrease of the
effective investors' risk aversion and other from the corresponding increase
on the spot price. The total effect mainly depends on the level of the agents'
risk aversions and the (uneven) production levels at initial and terminal time 
(see Figures \ref{fig:BM-3} and \ref{fig:BM-6}).

On the other hand, the theory of normal backwardation (see the seminal works by
\citet{Keynes30} and \citet{Hicks39}), states that there is a positive premium
required by the investors in order to satisfy the producers' hedging demand in
forward contracts. This premium, usually called \textit{forward} or 
\textit{insurance premium}, is given as the percentage difference between the 
expected commodity price at maturity and its forward price. As expected, this 
premium is increasing (decreasing) with respect to investors' (producers') risk 
aversion.

In contrast to the existing literature, our model includes as an input the
correlation between the stock market and the commodity demand shock. Several
empirical studies have shown that this correlation is indeed non-zero and, as
our results demonstrate, it does influence the equilibrium prices heavily. In
particular, as it is shown in Section \ref{sec:illu}, the effective
investors' risk aversion coefficient is decreasing in the presence of
non-zero correlation. This simply reflects the fact that higher correlation
means better hedging of the forward contract position by trading in the stock
market (provided there are no short-selling constraints on the investors'
trading strategies). Hence, non-zero correlation has in principle the same
effect on the equilibrium as a decrease in the investors' risk aversion (see
Figures \ref{fig:BM-1}, \ref{fig:BM-2} and \ref{fig:LJD-1}). For instance,
higher correlation (in absolute values) means higher spot commodity price, a
result that is also consistent with the observed market data (see e.g.
\cite{TangXio2012}). A similar effect is caused by an increased variance of the
demand shock, which can be due to the presence of jumps (see Figure
\ref{fig:LJD-1}).

\subsection{Relation with the existing literature}

Equilibrium pricing models in markets that consist of utility maximizing agents
have been recently addressed by a number of authors in mathematical finance;
see, among others, \citet{Anthropelos_Zitkovic:2010,Barrieu_ElKaroui_2009},
\citet{Cheridito_Horst_Kupper_Pirvu:2011}, \citet{Filipovic_Kupper_2008b,HorMul07} and
\citet{KarLehShre90}. The results in this literature however do not cover the
case of commodity forward contracts, not only because a commodity has a
consumption value which is reflected by the consumers' demand function, but
also due to the producers' specific storage choice. To the best of our
knowledge, this paper is the first to apply a utility maximization criterion
for spot and forward equilibrium prices of commodities, while considering also
the existence of a correlated stock market.

Theoretical studies of the equilibrium relationship between spot and forward
commodity prices go back to \citet{Stoll79,AndDan83} and
\citet{Hirsh88,Hirsh89}. The results of these seminal works are limited
regarding the agents' risk preferences, which are assumed to be mean-variance,
while recent extensions of this setting have followed different approaches than
ours. For instance in \citet{Baker12}, mean-variance optimization problems are
imposed in a discrete time dynamic model, where investors are the ones that
have the storage option and the consumers (the households) get utility from
consumption and the wealth (num\'{e}raire units). In \citet{RouSeppSpat00} and
\citet{Pirrong11}, investors are assumed to be risk neutral and without
access to other financial markets, while forward prices are simply the
expectations of future spot prices. The interaction between the optimal storage
and the investors' optimal position in the forward contract and its effect to
spot and forward equilibrium prices are also studied in \citet{EkeLauVil2014}.
However, in contrast to our model the investors trade only in forward
contracts, while the preferences are mean-variance, which means that they are
not monotonic with respect to futures revenues. More recently, endogenous
commodity supply under asymmetric information and limited participation has been
developed in \citet{LeclPraz14}. Static mean-variance models have been also
studied and statistically tested in \citet{AchaLochRama13} and
\citet{GorHayRou12}, however neither the investors nor the producers trade in
any other market outside of the commodity market\footnote{In
\cite{AchaLochRama13} investors are assumed risk neutral, but the imposed
capital constraints eventually lead to a mean-variance optimization criterion,
while in \cite{GorHayRou12} there is a random supply shock at the terminal
time, which however does not change the general idea of the equilibrium
setting.}. Hence, their theoretical results cover only a very special case of
our model, namely, when the stock and the commodity market are uncorrelated
and the demand random shift is normally distributed.\footnote{Continuous time
dynamic models with random demand shocks and exogenously given spot prices have
been developed in \cite{BasakPavlova13} and \cite{CasCollRou08}.}

The main novelties of our approach compared to the related literature are the
consideration of an exogenous stock market available in the investors' trading
set, the risk aversion of the agents' preferences and the much richer family of
processes that model the market factors.\footnote{The seminal works
\cite{Stoll79} and \cite{Hirsh88} also include a correlated risky asset in the
investors' set of strategies, forming an equilibrium framework. However, our
results are more general regarding not only the utility preferences and the
stochasticity of the market model, but also the set of investors' trading
strategies.} Indeed, as has already been discussed, both the correlation between
the stock and commodity market and the jump component do
influence the equilibrium prices.

% \bigskip

This paper is structured as follows: Section \ref{sec:model} sets up the
general framework for our equilibrium model. The well-posedness of the agents'
optimization problems and the existence of an equilibrium are proved in Section
\ref{sec:main}. Section \ref{model-ctdm} studies a model with continuous trading
under \lev dynamics, where semi-explicit formulas for equilibrium quantities are
derived and discussed. Finally, Section \ref{sec:illu} focuses on two examples
that permit the illustration and a further economic interpretation of the
results. Technical proofs of Section \ref{sec:illu} are placed in Appendix
\ref{app-proofs}.

\section{A general framework for commodity prices}
\label{sec:model}

We start by describing a general modeling framework where the interaction of 
market participants determines the spot and forward prices of commodities. The 
model consists of a pair of representative agents\footnote{The representation by 
a unique agent is widely used in this literature, see 
\cite{AchaLochRama13,BasakPavlova13,EkeLauVil2014,GorHayRou12, Singleton14} 
among others.}: the \textit{producers} produce the commodity, supply part of the 
production at the spot market and store the rest, while they hedge their 
exposure to price fluctuations using forward contracts on the commodity. The 
\textit{investors} invest in financial markets and, in order to diversify their 
portfolio, they also invest in the commodities forward market. Moreover, the 
model includes \textit{consumers} who consume the commodity at the spot market. 
The goal is to determine the price of the commodity that makes the forward 
market clear out, assuming that both producers and investors are
utility maximizers.

More specifically, the producers produce $\pi_0$ units of the commodity at the 
initial time 0 and $\pi_T$ units at the terminal time $T$; both $\pi_0$ and 
$\pi_T$ are assumed to be deterministic\footnote{This assumption means that the 
producers control the supply of their commodity only through the inventory 
management and not by changing their production plans, which can be 
prohibitively costly in the short-term (see also the related comment in 
\cite{AchaLochRama13}).}. They offer $\pi_0-\alpha$ units at the spot market at 
time $0$ and store the rest for time $T$. Furthermore, they hedge their exposure 
by investing in the forward market. Therefore, their position at time $T$ is
\begin{align}\label{eq:producers_position}
\underline{w}(\alpha,h^p)
 &= P_0(\pi_0-\alpha)(1+R)+P_T(\pi_T+\alpha(1-\varepsilon))+h^p(P_T-F),
\end{align}
where $P_0$ and $P_T$ denote the spot price at times $0$ and $T$ respectively, 
$R$ the discretely compounded interest rate, $\varepsilon\in[0,1]$ the cost of 
storage considered as percentage of the stored units\footnote{The representation 
of the storage cost as percentage is common in the related literature, see e.g. 
\cite{AchaLochRama13} and \cite{GorHayRou12}. The constant cost rate 
$\varepsilon$ is usually referred to as the \textit{depreciation rate}.}, $F$ 
the forward price and $h^p$ the amount of forward contracts held by the 
producers. A positive $h^p$ indicates a long position in the forward contract, 
while a negative $h^p$ amounts to a short one. The producers' utility is assumed 
to be exponential, henceforth their preferences are described by
\begin{align}
\UU_p(v)=-\frac{1}{\gamma_p}\log \E\sbrac{\e^{-\gamma_p v}},
\end{align}
where $\gamma_p>0$. As in \citet{AndDan80,AndDan83}, their problem is to find 
an optimal storage strategy $\alpha\in[0,\pi_0]$ and an optimal hedging strategy 
$h^p\in\R$ that maximize the utility of their position 
\eqref{eq:producers_position}. Therefore, their utility maximization problem is
\begin{equation}\label{eq:producers_problem}
\Pi^p
:= \underset{\alpha\in[0,\pi_0], \, h^p\in\R}{\sup}
    \UU_p\big(\underline{w}(\alpha,h^p)\big).
\end{equation}

The spot price of the commodity is the price at which the consumers' demand 
equals the producers' supply. The consumers' demand at the initial time is given 
by a strictly decreasing and linear function\footnote{The linearity of the 
demand function is imposed to facilitate the analysis. The limitation of this 
assumption does not exclude from our study the main characteristics of the 
demand, namely its elasticity and its random nature at terminal time. Let us 
also mention, that for a short time horizon a first order approximation of the 
demand function should suffice (see also the related discussion in 
\citep{AchaLochRama13, EkeLauVil2014}).}
\begin{align}\label{eq:def_lin-dem}
\psi_0(x) = \mu - mx,
\end{align}
where $\mu\in\R$ and $m\in\R_+$, while $x$ denotes the price. The parameter $m$ is a measure of the 
elasticity of demand for the commodity. The demand at the terminal time is 
random and depends on the factors driving the commodities market, which are 
incorporated in a random variable $X$. The demand function at the terminal time 
is of the form
\begin{equation}\label{eq:demand}
\psi_T(x) = \psi_0(x) + X.
\end{equation}
In other words, we assume that the shape and the elasticity of the demand 
function remain the same, however there is a random shift\footnote{A similar 
random shift has already been used in the literature, see for instance 
\cite{Pirrong11}.} acting on it. This shift may be, for example, the result of 
an increase or decrease in the prices of the competitive commodities, of 
fluctuations in a dominated currency, or of an exogenous increase in the demand 
for every price level. Since the demand function is linear, the inverse demand 
function is also linear and equals
\begin{align}\label{eq:inv_lin_dem}
\phi_0(y) = \frac{\mu-y}{m}
 \quad\text{ and }\quad
\phi_T(y) = \frac{\mu+X-y}{m}.
\end{align}
Henceforth, if the producers store $\alpha$ units at the initial time, the spot 
price of the commodity, determined by the equilibrium condition between demand 
and supply, equals
\begin{align}\label{eq:ini_price}
P_0 &= \phi_0(\pi_0-\alpha)
     = \phi_0(\pi_0)+\frac{\alpha}{m},
\end{align}
while the commodity spot price at the terminal time is
\begin{align}\label{eq:fin_price}
P_T &= \phi_T\big(\pi_T+\alpha(1-\varepsilon)\big)
     = \phi_0(\pi_T)-\frac{\alpha(1-\varepsilon)}{m}+\frac{X}{m}.
\end{align}
The producers control the spot price by choosing the inventory policy. By 
storing more commodity units they increase the spot price, but they also 
increase their exposure to the variation of the future spot price since the 
stored units will be supplied at the next time period.

The investors take a position $h^s$ in the forward contract and invest in an 
exogenously\footnote{In other words, the investors are price-takers when they 
invest in the financial market.} priced financial market. Their position at time 
$T$ equals
\begin{align}\label{eq:investors_position}
\overline{w}(G,h^s) &= h^s(P_T-F) + G,
\end{align}
for $G\in\mathcal{G}$, where $\mathcal{G}$ is a set of random variables that 
models discounted trading outcomes attainable with zero initial wealth. This 
general formulation allows to consider different scenarios simultaneously.

\begin{example}
The simplest scenario is $\mathcal{G}=\{0\}$, whence the investors can only 
invest in the forward contract. Another scenario is to consider an asset price 
process $S$ and denote by $G(\theta) = \int_0^\cdot \theta_u\ud S_u$ the gains 
process for a trading strategy $\theta$. In that case, the set of trading 
outcomes $\mathcal{G}$ is given by
\begin{align*}
\mathcal{G} = \{G_T(\theta): \theta \in \Theta \},
\end{align*}
for a set $\Theta$ of admissible, self-financing trading strategies. 
Transaction costs can be easily incorporated as well by setting
\begin{align*}
\mathcal{G} = \{G_T(\theta)-k(\theta): \theta \in \Theta \},
\end{align*}
where $k:\Theta\to\R$ is a concave function.
\end{example}

\noindent We assume that the investors' utility is also exponential with 
$\gamma_s>0$, that is, their preferences are described by
\begin{align}\label{eq:investors_utility}
 \UU_s (v)= -\frac{1}{\gamma_s} \log \E\sbrac{\e^{-\gamma_s v}},
\end{align}
therefore their utility maximization problem reads as
\begin{equation}\label{eq:investors_problem}
\Pi^s := \sup_{h^s\in\R,\,G\in\mathcal{G}} \UU_s\big(h^s(P_T-F) + G\big).
\end{equation}

The maximization problem of both participants depends on the forward price $F$. 
This price is determined by the equilibrium in the forward market, which is 
defined below.

\begin{definition}\label{def:static_equilibrium}
A triplet $(\hat\alpha,\hat{h},\hat{F})$ is called an \textit{equilibrium} if
it satisfies the following conditions:
\begin{itemize}
\item \textit{Market clearing}: the forward market clears out in the sense that
  \begin{equation}\label{eq:static_equilibrium_condition}
    \hat{h} := h^p(\hat{F}) = - h^s(\hat{F}).
  \end{equation}
\item \textit{Optimality}: the pair $(\hat\alpha,\hat{h})$ is optimal for the
      producers' problem $\Pi^p$ and $\hat{h}$ is optimal for the investors'
      problem $\Pi^s$.
\end{itemize}
The price $\hat{F}=F(\hat{\alpha})$ is called the \textit{equilibrium commodity
forward price} at maturity $T$. The induced price $\hat{P}_0:=P_0(\hat\alpha)$
derived by \eqref{eq:ini_price} is called the \textit{equilibrium commodity
spot price} at $0$.
\end{definition}

\begin{remark}
The utility maximization problems of both agents are equivalent to risk 
minimization problems relative to the entropic risk measure; see e.g. 
\citet{Barrieu_ElKaroui_2009}. The risk measure point of view is more natural 
for certain agents, such as a corporation managing its risk exposure.
\end{remark}

\section{Equilibrium in the general framework}
\label{sec:main}

The aim of this section is to show that an equilibrium exists in the general
modeling framework described above, under mild assumptions on the random
variable $X$ and the set of trading outcomes $\mathcal{G}$. Let $(\Omega,{\cal
F},\mathbb{P})$ be a probability space where $\F=\F_T$. In the sequel all
equalities and inequalities between random variables are understood in the
$\mathbb{P}$-almost sure sense. The interior and the boundary of a set $K$ are
denoted by $K^\circ$ and $\partial K$, respectively, and the domain of a
function $f$ by $\mathrm{dom}f$.

We denote the set of exponential moments of $X$ by $\UUU_X=\{u\in\R: \E[\e^{u
X}]<\infty\}$ and define the cumulant generating function of $X$ by
\begin{align}\label{defn:cgf}
\kappa_X(u)= \log\E\left[\e^{u X}\right],\qquad u\in\UUU_X.
\end{align}
The following conditions will be used throughout this work:
\begin{enumerate}[label={$(\mathbb{EM})$},leftmargin=40pt]
\item $0\in\UUU^\circ_X.$ \label{cond:EM}
\end{enumerate}
\begin{enumerate}[label={$(\mathbb{COE})$},leftmargin=40pt]
\item \textit{If $\partial\UUU_X=\pm\infty$ then the following limit holds:}
      \label{cond:COE}
      \begin{align*}
      \lim_{z\to\pm\infty}\frac{\kappa_X(z)}{|z|} =+\infty.
      \end{align*}
\end{enumerate}

\noindent The next lemma summarizes some useful properties of the cumulant
generating function.
\begin{lemma}\label{lem:kappa-properties}
The cumulant generating function $\kappa_X$ is convex and lower semicontinuous.
\end{lemma}

\begin{proof}
Convexity follows directly from H\"older's inequality; for $p,q\in(0,1)$
conjugate, we have that
\begin{align*}
\kappa_X(pu+qv)
 &= \log \E\big[ \e^{p u X} \e^{q v X} \big] \\
 &\le \log \big\{ \big(\E\big[\e^{uX}]\big)^p
                   \big(\E\big[\e^{vX}]\big)^q \big\}
  = p \kappa_X(u) + q \kappa_X(v).
\end{align*}
In order to show lower semicontinuity, consider a sequence $u_n\rightarrow u$;
then $\e^{u_nx}$ is a sequence of positive functions. Applying Fatou's lemma,
we get
\begin{align*}
\liminf_{u_n\to u} \kappa_X(u_n)
 &= \liminf_{u_n\to u} \log\E\big[\e^{u_n X}\big] \\
 &\ge \log\E\big[\liminf_{u_n\to u} \e^{u_n X}\big]
  = \kappa_X(u).
\end{align*}
\end{proof}

\subsection{Producers' optimization problem}

The first step is to consider the producers' optimization problem and show that 
it admits a maximizer under mild assumptions. The producers' position, using the 
spot market equilibrium conditions \eqref{eq:ini_price} and 
\eqref{eq:fin_price}, can be written as
\begin{align}\label{eq:pop-1}
\underline{w}(\alpha,h^p)
 &= P_0(\pi_0-\alpha)(1+R)+P_T(\pi_T+\alpha(1-\varepsilon))+h^p(P_T-F)
\nonumber\\
 &\stackrel[\eqref{eq:fin_price}]{\eqref{eq:ini_price}}{=}
    \left(\phi_0(\pi_0)+\frac{\alpha}{m}\right)(\pi_0-\alpha)(1+R)
  - h^pF \nonumber\\
 &\quad + \big(\pi_T+\alpha(1-\varepsilon)+h^p\big)
   \left(\phi_0(\pi_T)-\frac{\alpha(1-\varepsilon)}{m}+\frac{X}{m}\right)
\nonumber\\
 &=: \QQQ(\alpha,h^p) + \ell(\alpha,h^p) X,
\end{align}
where $\QQQ$ is a quadratic function\footnote{It follows directly from 
\eqref{eq:pop-1} (see also the associated formulas in subsection \ref{subsec:bm 
model}), that if the parameter $\mu$ is sufficiently small, then producers may 
have motive to discard the commodity, in the sense that the total optimal supply 
is less than the total production, even if the demand function is deterministic. 
We can avoid such cases, by assuming that the parameter $\mu$ is sufficiently 
large. Note that $\mu$ could be considered as the consumers' demand when the 
commodity has zero price, hence assuming large values for $\mu$ is a reasonable 
assumption.} in $\alpha$ and $h^p$ of the form
\begin{align}\label{eq:QQQ}
\QQQ(\alpha,h^p)
&= - \alpha^2 \frac{1+R+(1-\varepsilon)^2}{m}
   + \alpha \frac{2(1+R)\pi_0-2(1-\varepsilon)\pi_T-(R+\varepsilon)\mu}{m}
\nonumber\\
&\quad - \alpha h^p\frac{1-\varepsilon}{m} -
h^p\left(F-\frac{\mu-\pi_T}{m}\right)
    + \pi_T \phi_0(\pi_T) + \pi_0\phi_0(\pi_0)(1+R), 
\end{align}
while $\ell$ is a bilinear function in $\alpha$ and $h^p$ given by
\begin{equation}\label{eq:ell}
\ell(\alpha,h^p) = \frac{\alpha(1-\varepsilon)+h^p+\pi_T}{m}.
\end{equation}

Using the translation invariance of the exponential utility function, the 
producers' utility takes the form
\begin{align}\label{eq:prods-utility}
\UU_p\big(\underline{w}(\alpha,h^p)\big)
 &= -\frac{1}{\gamma_p} \log \EE\Big[\exp \Big( -\gamma_p \big\{\QQQ(\alpha,h^p)
      + \ell(\alpha,h^p) X \big\} \Big)\Big] \nonumber\\
 &= \QQQ(\alpha,h^p) - \frac{1}{\gamma_p}\log
    \EE\Big[ \exp\Big( -\gamma_p \ell(\alpha,h^p) X\Big)\Big] \nonumber\\
 &= \QQQ(\alpha,h^p) - \frac{1}{\gamma_p}
    \kappa_X\big( -\gamma_p \ell(\alpha,h^p) \big),
\end{align}
assuming that $-\gamma_p\ell(\alpha,h^p)\in\UUU_X$. In the sequel, we will work 
with the extended producers' utility
$\widetilde\UU_p(\underline{w}(\alpha,h^p))$ which is defined as follows:
\begin{align}\label{eq:extended_prod_utility}
\widetilde\UU_p\big(\underline{w}(\alpha,h^p)\big)
=\left\{%
\begin{array}{ll}
  \UU_p\big(\underline{w}(\alpha,h^p)\big), &  \text{if }
(\alpha,h^p)\in\widetilde{\UUU}_X,\\
 -\infty, & \text{otherwise}, \\
\end{array}%
\right.
\end{align}
where $\widetilde{\UUU}_X = \cbrac{(x_1,x_2)\in\R^2: 
-\gamma_p\ell(x_1,x_2)\in\UUU_X }$. The producers' optimization problem 
\eqref{eq:producers_problem} can then be written as follows
\begin{align}\label{eq:producers_problem_conju}
\Pi^p &= \underset{\alpha\in[0,\pi_0]}{\sup}
         \underset{h^p\in\R}{\sup}
         \widetilde\UU_p\big(\underline{w}(\alpha,h^p)\big) \nonumber\\
      &= \underset{\alpha\in[0,\pi_0]}{\sup}
         \underset{h^p\in\R}{\sup}\{u_p(\alpha,h^p)-h^pF\}
       = \underset{\alpha\in[0,\pi_0]}{\sup} \{ -u_p^*(\alpha,F) \},
\end{align}
where
\begin{align}\label{eq:def-up}
u_p(\alpha,h^p)
=\left\{%
 \begin{array}{ll}
 \QQQ(\alpha,0) - \frac{1}{\gamma_p}
  \kappa_X\big( -\gamma_p \ell(\alpha,h^p) \big)
 - h^p \ell(\alpha,-\mu),
 & \hbox{if $(\alpha,h^p)\in\widetilde{\UUU}_X$,} \\
  -\infty, & \hbox{otherwise,} \\
\end{array}%
\right.
\end{align}
while $u_p^*(\alpha,\cdot)$ denotes the conjugate function of
$u_p(\alpha,\cdot)$, for every $\alpha\in[0,\pi_0]$.

\begin{proposition}\label{pro:producers_problem}
Assume that conditions \ref{cond:EM} and \ref{cond:COE} hold. Then, for every 
$F\in\R$ there exists a maximizer $(\hat{\alpha},\hat{h}^p)$ for the producers' 
problem $\Pi^p$ such that $(\hat{\alpha},\hat{h}^p)\in\widetilde{\UUU}_X$.
\end{proposition}

\begin{proof}
The function $\widetilde\UU_p(\underline{w}(\alpha,h^p))$ in 
\eqref{eq:prods-utility} is upper semicontinuous, strictly concave in $\alpha$ 
and concave in $h^p$, since $\QQQ$ is quadratic, $\ell$ is linear and $\kappa_X$ 
is convex and lower semicontinuous in its arguments; see Lemma 
\ref{lem:kappa-properties} and \eqref{eq:QQQ}--\eqref{eq:ell}. Observe that 
$\alpha$ takes values in a bounded set. If the set $\mathcal{U}_X$ is also 
bounded, then the existence of a maximizer follows by the concavity and the 
upper semicontinuity of $\UU_p(\underline{w}(a,h^p))$. Otherwise, if 
$\mathcal{U}_X$ is unbounded, using Assumption \ref{cond:COE}, the linearity of 
$\ell$ in $\alpha$ and that $\alpha$ belongs to a bounded set, we get that
\begin{align}
\lim_{h^p\to\pm\infty} \inf_{\alpha\in[0,\pi_0]}
 \frac{\kappa_X\big(-\gamma_p \ell(\alpha,h^p)\big)}{|h^p|} = +\infty.
\end{align}
Therefore, $\UU_p(\underline{w}(\alpha,h^p))$ is coercive in $h^p$ resulting in 
the existence of a maximizer. Finally, if $(\hat{\alpha},\hat{h}^p)$ does not 
belong to $\widetilde{\UUU}_X$, then the utility of the producers is not 
maximized, see \eqref{eq:extended_prod_utility}.
\end{proof}

\begin{corollary}\label{corr:spec_fun_properties}
Assume that conditions \ref{cond:EM} and \ref{cond:COE} hold. Then, the
function $u_p(\alpha,\cdot)$ is concave and upper semicontinuous for every
$\alpha\in[0,\pi_0]$. In addition, it is coercive uniformly in $\alpha$, that
is
\begin{align}\label{eq:spec_unif_coercivity}
\lim_{h^p\to\pm\infty} \sup_{\alpha\in[0,\pi_0]} \frac{u_p(\alpha,h^p)}{|h^p|}
= - \infty.
\end{align}
\end{corollary}

\begin{remark}\label{rem: storage}
It follows from \eqref{eq:QQQ}, that small production at time $T$ raises the 
producers' desire to store, even when the future demand function is 
deterministic. This occurs because a possible scarcity of the commodity at time 
$T$ would result in higher future spot prices, hence producers would be better 
off storing some production and selling it at time $T$. On the other hand, 
higher future production decreases the optimal storage choice. Hence, the 
producers' desire to balance uneven productions is an important feature that 
influences the optimal storage choice.
\end{remark}

\subsection{Investors' optimization problem}

The second step is to analyze the structure and properties of the investors' 
optimization problem. Although we cannot prove the existence of a maximizer at 
this level of generality, the results we obtain are sufficient to show the 
existence of an equilibrium in the next subsection.

Let $\Q$ be a probability measure on $(\Omega,\F)$. The relative entropy 
$\mathcal{H}(\Q|\P)$ of $\Q$ with respect to $\P$ is defined by
\begin{align*}
\mathcal{H}(\QQ|\PP)
 =\left\{%
  \begin{array}{ll}
   \EE_{\QQ}\big[\ln\big( \frac{\ud\QQ}{\ud\PP}\big)\big],
        & \text{if } \QQ\ll\PP, \\
   +\infty, & \text{otherwise}. \\
\end{array}%
\right.
\end{align*}
Given $\alpha\in[0,\pi_0]$, the spot price of the commodity $P_T=P_T(\alpha)$ 
is provided by \eqref{eq:fin_price}. Define the function
\begin{align}\label{eq:def-us}
u_s(\alpha,h^s)
:=\sup_{G\in{\cal G}} \UU_s \big( h^s P_T + G \big),
\end{align}
for a convex set ${\cal G}$ of $\F_T$-measurable random variables that contains 
$0$. In order to prove the existence of an equilibrium we will make use of the 
following assumption:
\begin{enumerate}[label={$(\mathbb{USC})$},leftmargin=40pt]
\item The function $h^s\mapsto u_s(\alpha,h^s)$ is upper semicontinuous for
      every $\alpha\in[0,\pi_0]$. \label{cond:USC}
\end{enumerate}
The function $h^s\mapsto u_s(\alpha,h^s)$ is also concave for every 
$\alpha\in[0,\pi_0]$, while the investors' optimization problem can be
expressed as follows
\begin{align}\label{eq:specs_prob-1}
\Pi^s
 = \sup_{h^s\in\mathbb{R}}\sup_{G\in{\cal G}}
   \UU_s \big( h^s (P_T-F) + G\big)
 = \sup_{h^s\in\mathbb{R}} \left\{ u_s(\alpha,h^s)-h^s F \right\}.
\end{align}

\noindent Throughout this section, we will also make use of the sets
\[ \MM_{\cal G} := \big\{\QQ\ll\P :\mathcal{H}(\QQ|\PP)<\infty\ \text{ and }\
  \E_{\mathbb Q}[G]\leq 0\ \text{ for all }G\in{\cal G} \big\} \]
and
\[ {\cal Q }_X := \big\{\QQ\ll\P :\E_{\mathbb{Q}}[|X|]<\infty\}.\]
The financial market is free of arbitrage if $\MM_{\cal G}\neq\emptyset$. This 
is a sufficient condition, but not necessary, since it also requires the entropy 
to be finite. In the sequel, we also need the existence of at least one 
probability measure in $\MM_{\cal G}$ that belongs to ${\cal Q }_X$\footnote{As 
we will see later on, this assumption is needed in order to guarantee that the 
commodity spot price $P_T(\alpha)\in L^1(\QQ)$ for at least one $\QQ\in\MM_{\cal 
G}$ and $\alpha\in[0,\pi_0]$, which eventually implies that the investors' 
utility is bounded from above. In the market model of Section \ref{model-ctdm}, 
this assumption implies, in particular, that the investor's indifference price 
of the commodity is bounded from above; see also Remark \ref{rem:NA 
condition}.}. We state these requirements in the following condition:
\begin{enumerate}[label={$(\mathbb{NA})$},leftmargin=40pt]
\item $\MM_{\cal G}\cap{\cal Q}_X\neq\emptyset$.\label{cond:NA}
\end{enumerate}

\begin{proposition}\label{pro:investors_problem}
Assume that \ref{cond:NA} holds. Then, for each $\alpha\in[0,\pi_0]$ there
exists $F=F(\alpha)\in \mathbb{R}$ such that
\begin{equation}\label{r1}
\limsup_{h^s\to\pm\infty} \frac{u_s(\alpha,h^s)}{|h^s|}<+\infty,
\end{equation}
and
\begin{equation}\label{r2}
-u^\ast_s(\alpha,F)
 := \sup_{h^s\in\R}\left\{ u_s(\alpha,h^s)-h^s F\right\} < +\infty.
\end{equation}
\end{proposition}

\begin{proof}
Fix $\QQ\in\MM_{\cal G}\cap{\cal Q }_X$. Using \ref{cond:NA} and 
\eqref{eq:fin_price} we get that $P_T=P_T(\alpha)\in L^1(\QQ)$ for all
$\alpha\in[0,\pi_0]$. According to \citet[Lemma 3.29]{Foellmer_Schied_2004}, 
for each $G\in{\cal G}$, $h^s\in\mathbb{R}$ and $n\in\mathbb{N}$ it holds that
\begin{align}
\UU_s\big( [h^s P_T + G] \vee (-n) \big)
 &= -\frac{1}{\gamma_s}\log \E\big[ \exp \big\{
   - \gamma_s \big( [h^s P_T + G] \vee (-n) \big)\big\}\big] \nonumber\\
 &\le \E_{\mathbb{Q}} \big[ (h^s P_T + G )\vee (-n) \big]
   + \frac{1}{\gamma_s}{\cal H}(\mathbb{Q}|\mathbb{P}).
\end{align}
Since $(P_T+G)^+\in L^1(\mathbb{Q})$ and $\E_{\mathbb{Q}}[G]\le 0$, monotone
convergence implies that
\begin{align*}
u_s(\alpha,h^s)
 = \sup_{G\in {\cal G}} \Big\{ -\frac{1}{\gamma_s} \log \E \big[
    \exp\big\{ -\gamma_s (h^s P_T + G) \big\}\big] \Big\}
\le h^s \E_{\mathbb{Q}}[P_T]
  + \frac{1}{\gamma_s}\mathcal{H}(\mathbb{Q}|\mathbb{P}),\label{bound:us}
\end{align*}
which yields \eqref{r1}. Finally, defining $F:= \E_{\mathbb{Q}}[P_T]$ we obtain
that
\begin{align}\label{boundH}
\sup_{h^s\in\mathbb{R}} \big\{u_s(\alpha,h^s)-h^s F \big\}
  \le \frac{1}{\gamma_s}{\cal H}(\mathbb{Q}|\mathbb{P})<+\infty.
\end{align}
\end{proof}

\subsection{Existence of equilibrium}

We are now ready to show that under mild assumptions an equilibrium exists in 
the general modeling framework described in Section \ref{sec:model}. Explicit, 
and easily verifiable, conditions for the uniqueness of the equilibrium are also 
provided. We start with some preparatory results from convex analysis before 
stating and proving the main theorem.

According to \eqref{eq:producers_problem_conju}, the producers' optimization
problem is described by
\begin{align}\label{op1}
\Pi^p
 &= \sup_{\alpha\in[0,\pi_0]}\sup_{h^p\in\R}
    \left\{u_p(\alpha,h^p)-h^p F\right\}
  =-\inf_{\alpha\in[0,\pi_0]} \inf_{h^p\in\R}
    \left\{ h^p F - u_p(\alpha,h^p)\right\} \notag\\
 &= \sup_{\alpha\in[0,\pi_0]} \{ -u_p^\ast(\alpha,F) \}.
\end{align}
Similarly, from \eqref{eq:specs_prob-1} and \eqref{r2} the investors'
optimization problem is described by
\begin{align}\label{op2}
\Pi^s &= \sup_{h^s\in\R} \left\{u_s(\alpha,h^s)-h^s F\right\}
       =-u_s^\ast(\alpha,F).
\end{align}
In the sequel we will make use of several results from convex analysis; we
refer the reader to \citet{Rockafellar_1997} for a comprehensive introduction.
We define the sup-convolution of $u_p$ and $u_s$ via
\begin{equation}\label{convolution-definition}
u(\alpha,h):=\sup_{h^p+h^s=h} \left\{u_p(\alpha,h^p)+u_s(\alpha,h^s)\right\},
\end{equation}
and we know that its conjugate function satisfies
\begin{align}\label{convolution-property}
u^*(\alpha,F)
 &= \inf_{h\in\mathbb{R}} \{hF-u(\alpha,h)\}
  = u_p^*(\alpha,F) + u_s^*(\alpha,F);
\end{align}
cf. \cite[Theorem 16.4]{Rockafellar_1997}. Moreover, it holds that
\begin{align}\label{convolution-simple}
u(\alpha,h) = \inf_{F\in\mathbb{R}} \{hF-u^*(\alpha,F)\}
\end{align}
and we know that $F$ belongs to the supergradient of $u(\alpha,h)$, denoted by
$\partial u(\alpha,h)$, if the equality
\begin{align}\label{supergradient}
u(\alpha,h) = hF-u^*(\alpha,F)
\end{align}
is satisfied; see \cite[Theorem 23.5]{Rockafellar_1997}

\begin{theorem}\label{thm:main}
Assume that conditions \ref{cond:EM}, \ref{cond:COE}, \ref{cond:USC} and
\ref{cond:NA} hold, and suppose that
\begin{equation}\label{cond:boundedF}
-\gamma_p\ell(\pi_0,0) \in \UUU^\circ_X.
\end{equation}
Then there exists an equilibrium $(\hat\alpha,\hat{h},\hat{F})$.
\end{theorem}

\begin{remark}[Uniqueness]\label{rem:uniq}
The equilibrium in Theorem \ref{thm:main} is not unique in general, since the 
supergradient $\partial u(\alpha,0)$ is not a singleton. However, if the 
functions $h^p\mapsto u_p(\alpha,h^p)$ and $h^s\mapsto u_s(\alpha,h^s)$ are 
differentiable for all $\alpha\in[0,\pi_0]$, then the equilibrium commodity 
forward price is unique. Indeed, the proof of Theorem \ref{thm:main} yields that 
any equilibrium commodity forward price $\hat F$ satisfies $\hat F\in \partial 
u(\hat\alpha,0)$ for the unique optimizer $\hat \alpha\in[0,\pi_0]$. If 
$u_p(\hat\alpha,\cdot)$ and $u_s(\hat\alpha,\cdot)$ are both differentiable then 
it follows, for instance from Lemma 1.6.5 in \citet{Cheridito_2013}, that 
$u(\hat\alpha,\cdot)$ is differentiable at $0$, in which case  $\partial 
u(\hat\alpha,0)$ is a singleton. Moreover, if $h^p\mapsto u_p(\hat\alpha,h^p)$ 
and $h^s\mapsto u_s(\hat\alpha,h^s)$ are strictly concave, then also the optimal 
strategy $\hat h$ is unique. These conditions can be easily verified in the 
examples; see Sections \ref{model-ctdm} and \ref{sec:illu}.
\end{remark}

\begin{proof}
The proof of this theorem is carried out in three steps and the strategy is
represented by the following diagram:
\begin{center}
\begin{tikzpicture}
\matrix (m) [matrix of math nodes,row sep=4em,column sep=5em,minimum width=3em]
{
     \alpha^n & \alpha^{\phantom{n}} \\
     F(\alpha^n) = F^n & F = F(\alpha) \\};
  \path[-stealth]
    (m-1-1) edge node [left] {S1} (m-2-1)
            edge node [below] {}  (m-1-2)
    (m-2-1.east|-m-2-2) edge node [below] {S2} (m-2-2)
    (m-1-2) edge node [right] {S1} (m-2-2)
    (m-2-2) edge  [loop] node[left=5mm] {S3} (m-2-2);
\end{tikzpicture}
\end{center}
The first step is to show that for every fixed $\alpha$ there exists an
equilibrium. Then, we consider a sequence $(\alpha^n)$ maximizing the producers'
utility that converges to some $\alpha$. The previous step yields the existence
of equilibrium prices $F(\alpha^n)=F^n$ and $F(\alpha)$ corresponding to
$\alpha^n$ and $\alpha$, respectively. The second step is to show that the
equilibrium prices $F^n$ converge to some limit, denoted by $F$. The final step
is to show that $F$ equals $F(\alpha)$.

\textit{Step 1:} Fix $\alpha\in[0,\pi_0]$. According to Propositions
\ref{pro:producers_problem} and \ref{pro:investors_problem}, there exists a
price $F=F(\alpha)\in\mathbb{R}$ such that
\begin{equation}\label{boundabove}
u(\alpha,h) \le \sup_{h^p} \left\{u_p(\alpha,h^p)-h^p F\right\}
+ \sup_{h^s} \left\{u_s(\alpha,h^s)-h^s F\right\} + hF <\infty.
\end{equation}
Using \eqref{cond:boundedF} and conditions \ref{cond:EM} and \ref{cond:NA} we
get that $u_p(\alpha,\cdot)>-\infty$ and $u_s(\alpha,\cdot)>-\infty$ in a
neighborhood of $0$. Hence $u(\alpha,\cdot)>-\infty$ on a neighborhood of $0$,
therefore $0$ belongs to the interior of $\mathop{\rm dom} u(\alpha,\cdot)$,
which by \cite[Theorem 23.4]{Rockafellar_1997} implies that $\partial
u(\alpha,0) \neq \emptyset$. Let $F(\alpha)$ be an element of the supergradient
$\partial u(\alpha,0)$. Then
\begin{align}\label{conv2}
u(\alpha,0)
 &\le \sup_{h^p} \left\{u_p(\alpha,h^p)-h^p F(\alpha)\right\}
    + \sup_{h^s} \left\{u_s(\alpha,h^s)-h^s F(\alpha)\right\} \nonumber\\
 &= -u_p^\ast(\alpha,F(\alpha))-u_s^\ast(\alpha,F(\alpha)) \nonumber\\
 &= u(\alpha,0)= \sup_{h^p+h^s=0}
    \left\{u_p(\alpha,h^p)+u_s(\alpha,h^s)\right\},
\end{align}
where the second to last equality follows from \eqref{convolution-property} and
\eqref{supergradient} using that $h=0$. By means of Corollary
\ref{corr:spec_fun_properties} and Proposition \ref{pro:investors_problem} we
deduce that the function $h\mapsto u_p(\alpha,h)+u_s(\alpha,-h)$ is concave and
tends to $-\infty$ as $h\to\pm\infty$; see in particular
\eqref{eq:spec_unif_coercivity} and \eqref{r1}. Therefore, the supremum in
\eqref{conv2} is attained for $h^p(\alpha),h^s(\alpha)\in\mathbb{R}$ with
$h^p(\alpha)+h^s(\alpha)=0$. Moreover, it follows from \eqref{conv2} that
\[
h^p(\alpha)=\mathop{\rm argmax}\left\{ u_p(h^p,\alpha)-h^p F(\alpha)\right\}
  \quad\mbox{and}\quad
h^s(\alpha)=\mathop{\rm argmax}\left\{ u_s(h^s,\alpha)-h^s F(\alpha)\right\}.
\]
In other words, for every fixed $\alpha\in[0,\pi_0]$ there exists an
equilibrium.

\textit{Step 2:} Consider an optimizing sequence $(\alpha^n)$ for the
producers' utility converging to $\alpha$, then
\begin{equation}\label{eqmax}
-u_p^\ast(\alpha^n,F^n) \xrightarrow[n\to\infty]{}
 \sup_\alpha \{-u_p^\ast(\alpha,F(\alpha))\},
\end{equation}
where $F^n=F(\alpha^n)$ is the sequence of equilibrium prices corresponding to
$\alpha^n$. Let us now prove that both the equilibrium prices $F^n$ and the
optimal strategies $h^n=h^p(\alpha^n)=-h^s(\alpha^n)$ are bounded; henceforth
$h^n\to h$ and $F^n\to F$ by possibly passing to a subsequence.

The upper semicontinuity of $u_p$, condition \ref{cond:USC} and the definition
of the sup-convolution yield that
\begin{align*}
\limsup_{n\to+\infty} u(\alpha^n,0)
 \le \limsup_{n\to+\infty}\{u_p(\alpha^n,h^n)+u_s(\alpha^n,-h^n)\}
 \le u_p(\alpha,h)+u_s(\alpha,-h)
 \le u(\alpha,0),
\end{align*}
which is finite by \eqref{boundabove}. Moreover, due to condition
\eqref{cond:boundedF} there exists a neighborhood ${\cal V}$ of 0 such that
\begin{align*}
\inf_{h^p\in{\cal V}}\inf_{n\in\mathbb{N}} u_p(\alpha^n,h^p)>-\infty.
\end{align*}
Hence, there exist constants $c_1\in\mathbb{R}$ and $c_2>0$ such that
\begin{align*}
-u^\ast_p(\alpha^n,F^n)
 =\sup_{h^p\in\mathbb{R}}\left\{ u_p(\alpha^n,h^p)-F^n h^p\right\}
 \ge c_1+c_2|F^n|
\end{align*}
and similarly $-u^\ast_s(\alpha^n,F^n)\ge c_1+c_2|F^n|$. Therefore,
$2c_1+2 c_2|F^n|\le \sup_{n\in\mathbb{N}} u(\alpha^n,0)<+\infty$ showing
that $(F^n)$ is bounded. Since
\[-u^\ast_p(\alpha^n,F^n)=u_p(\alpha^n,h^n)-h^n F^n,\]
it follows from Corollary \ref{corr:spec_fun_properties} that $(h^n)$ is also
bounded.

\textit{Step 3:} Finally, the goal is to identify $F$ as the desired
equilibrium price, that is, prove that $F=F(\alpha)$. We start by showing that
\begin{align}\label{step1}
u_p^\ast(\alpha^n,F^n)\xrightarrow[n\to\infty]{} u^\ast_p(\alpha,F).
\end{align}
Indeed, by continuity of $u_p(\alpha,h^p)$ in $\alpha$ we get that
$h^pF^n-u_p(\alpha^n,h^p)\to h^p F-u_p(\alpha,h^p)$. Thus, by the definition of
the conjugate $u^\ast_p(\alpha,F)=\inf_{h^p}\{ h^p F-u_p(\alpha,h^p)\}$, it
follows that
\begin{align*}
\limsup_{n\to\infty} u^\ast_p(\alpha^n,F^n)\le u^\ast_p(\alpha,F).
\end{align*}
Moreover, equilibrium prices belonging to the supergradient of $u$, ensures that
\begin{align*}
\liminf_{n\to\infty} u^\ast_p(\alpha^n,F^n) 
	= \liminf_{n\to\infty}\{ h^n F^n-u_p(\alpha^n,h^n) \}
	\ge h F-u_p(\alpha,h)
	\ge u^\ast_p(\alpha,F)
\end{align*}
and thus $\liminf_{n\to\infty} u^\ast_p(\alpha^n,F^n)\ge u^\ast_p(\alpha,F)$.
Hence \eqref{step1} holds. The same argumentation implies that
$u^\ast_s(\alpha^n,F^n)\to u^\ast_s(\alpha,F)$.

Next, we show that $u(\alpha^n,0)\to u(\alpha,0)$. On the one hand, there
exists an $h'\in\mathbb{R}$ such that
\begin{align*}
u(\alpha,0)
 = u_p(\alpha,h')+u_s(\alpha,-h')
 = \lim_{n\to\infty} \{u_p(\alpha^n,h')+u_s(\alpha^n,-h')\}
 \le \liminf_{n\to\infty} u(\alpha^n,0).
\end{align*}
The first equality holds since the supremum is attained, the second follows
from the continuity of $u_p$ and $u_s$ in $\alpha$ and the last one by
\eqref{convolution-definition}. On the other hand, $(h^n)$ converging to $h$
implies
\begin{align*}
\limsup_{n\to\infty} u(\alpha^n,0)
  = \limsup_{n\to\infty} \{u_p(\alpha^n,h^n)+u_s(\alpha_n,-h^n)\}
  \le u_p(\alpha,h)+u_s(\alpha,-h)
  \le u(\alpha,0),
\end{align*}
making use of the same argumentation for each equality as above.

Summarizing, using the convergence of the sup-convolutions,
\eqref{supergradient}, \eqref{convolution-property} and the convergence of the
conjugates, we arrive at
\begin{equation}\label{eq:sum}
u(\alpha,0)
 = \lim_{n\to\infty} u(\alpha^n,0)
 = \lim_{n\to\infty} \{-u^\ast_p(\alpha^n,F^n)-u^\ast_s(\alpha^n,F^n)\}
 = -u^\ast_p(\alpha,F)-u^\ast_s(\alpha,F).
\end{equation}
In particular, the sup-convolution $u(\alpha,0)$ is attained at $h^p(\alpha)=h$
and $h^s(\alpha)=-h\in\mathbb{R}$ for which $h^p(\alpha)+h^s(\alpha)=0$.
Therefore, according to \eqref{eqmax} and \eqref{eq:sum} the pair
$(\alpha,h^p(\alpha))$ and $h^s(\alpha)$ are optimal trading strategies for the
price $F$ which satisfy the clearing condition. Hence $(\alpha,h^p(\alpha),F)$
is an equilibrium.
\end{proof}

\section{A model with continuous trading and dependent markets}
\label{model-ctdm}

In this section, we consider a model where investors are allowed to trade 
continuously over time in the financial market, while the dynamics of the 
financial and the commodity markets are dependent and driven by \lev processes. 
The aim is to derive explicit representations for the optimization problems of 
the producers and the investors. 

\lev processes have been used for modeling variables in finance, such as stocks 
or interest rates, whose return distributions exhibit fat tails and skew, 
because they can combine realistic features with analytical tractability; see 
e.g. \citet{Eberlein01a,Carretal02,ContTankov03} and \citet{Schoutens03}. 
\citet{GorRou06} provide evidence that commodity futures exhibit similar 
behavior. Using \lev processes, we can easily combine diffusions with jump 
processes, while different types of dependence structures can also be 
incorporated. 

In the model considered in this section, the investors observe the evolution of 
the consumers' demand through time and adjust their trading strategy 
dynamically\footnote{It is implicitly assumed that the investors' investment 
choices are independent of their possible commodity consumption policy. This 
assumption has been imposed in the majority of the related literature (cf. 
\cite{AchaLochRama13,CasCollRou08,EkeLauVil2014,Hirsh88}) and implies that the 
consumers or corporations that use the commodity to produce other goods are only 
a small part of the investors' side and their possible joint optimization 
problem is negligible when the investors are considered as a whole.}. Moreover, 
the uncertainty in the evolution of the consumers' demand and the evolution of 
the financial market are dependent processes which can exhibit `shocks' (i.e. 
large jumps). The producers are trading in the forward market only at discrete 
time instances, associated with their production schedule\footnote{The fact that 
dynamic trading of the commodity forward contract is not considered implies that 
producers counterpart the position of investors only at the initial and the 
terminal time for hedging purposes. During the time period $(0,T)$, investors 
may trade in the forward market and form their representative agent's position. 
The focus of our analysis is how the interaction of producers and investors 
results in the equilibrium prices at times 0 and $T$.}. This setting reflects 
real-world situations, in the sense that the arrival of certain news can affect 
both the demand for a certain commodity as well as the financial market, these 
processes are observable over time and investors typically trade continuously in 
the financial market and adjust their portfolios according to new information.

Consider a complete stochastic basis $(\Omega,\F,\bF,\P)$ where
$\bF=(\F_t)_{t\in[0,T]}$ denotes the filtration (flow of information). Let
$\prozess[Z]$ be an $\R^d$-valued \lev process with characteristic triplet
$(b,c,\nu)$, where $b\in\Rd$, $c$ is a symmetric, non-negative definite $d\times
d$ matrix and $\nu$ is a \lev measure; see e.g. \citet{Applebaum09,Kyprianou06}
or \citet{Sato99} for more details on \lev processes. Denote the set of
exponential moments of $Z_t$, $t\in[0,T]$, by
\begin{align}\label{def:set-U}
\UUU_Z
 = \bigg\{ u\in\Rd: \E\big[ \e^{\scal{u}{Z_t}} \big] <\infty \bigg\}
 = \bigg\{ u\in\Rd: \int_{|x|>1} \e^{\scal{u}{x}} \nu(\dx) <\infty \bigg\}.
\end{align}
This set is convex and contains the origin, cf. \citet[Thm.~25.17]{Sato99}.
Assuming that $0\in\UUU_Z^\circ$, exponential moments exist and the \lito
decomposition takes the form
\begin{align}
Z_t = bt + \sqrt{c} W_t + \int_0^t\int_{\R^d} x (\mu^Z-\nu^Z)\dsdx,
\end{align}
where $\mu^Z$ is the random measure of jumps of the process $Z$ with compensator
$\nu^Z=\text{Leb}\otimes\nu$. The moment generating function of $Z_t$ is
well-defined for every $u\in\mathcal{U}_Z$ and we know from the \lk formula that
\begin{align}\label{eq:Levy-cumulant}
 \E\big[ \e^{\scal{u}{Z_t}} \big] = \exp\big(t\kappa(u)\big),
\end{align}
where $\kappa$ denotes the cumulant generating function of $Z_1$, that is
\begin{align}\label{def:cgf}
\kappa(u) = \scal{u}{b} + \frac{\scal{u}{cu}}{2}
          + \int_{\R^d} (\e^{\scal{u}{x}} -1 - \scal{u}{x}) \nu(\dx).
\end{align}
Moreover, if $0\in\UUU_Z^\circ$, then the cumulant generating function $\kappa$
is real analytic in the interior of $\UUU$ and thus smooth; cf. \citet[Lemma
2.1]{Eberlein_Glau_2014}.

The uncertainty in the financial and the commodity markets is modeled using
the \lev process $Z$ and a factor structure. More precisely, we consider vectors
$u_1,u_2\in\R^d$ that specify how $Z$ influences each market. We will
incorporate the financial market in a representative stock index whose
discounted price process $S$ is modeled by
\begin{align}\label{eq:model-spec}
S_t = S_0 \e^{Y_t}
    \quad \text{where} \quad
Y_t = \scal{u_1}{Z_t},
\end{align}
with $S_0\in\R_+$ and $t\in[0,T]$. Moreover, the random variable $X$ that
determines the consumers' demand function at the terminal time is modeled via
\begin{align}\label{eq:model-prod}
 X = \scal{u_2}{Z_T}.
\end{align}

\subsection{The producers' optimization problem revisited}

The cumulant generating function of the random variable $X=\scal{u_2}{Z_T}$ in
this setting, using \eqref{eq:Levy-cumulant}, takes the form
\begin{align}\label{eq:kappa_X-2}
\kappa_X(v) = \kappa(vu_2) T =: \kappa_2(v) T,
\end{align}
and the set of exponential moments equals $\UUU_X=\{v\in\R: vu_2\in\UUU_Z\}$.
Therefore, the function $u_p$ in the producers' optimization problem
\eqref{eq:producers_problem_conju}--\eqref{eq:def-up} can be rewritten as
\begin{align}\label{eq:up-Levy}
u_p(\alpha,h^p)
=\left\{%
 \begin{array}{ll}
 \QQQ(\alpha,0) - \frac{1}{\gamma_p}
  \kappa_2\big( -\gamma_p \ell(\alpha,h^p) \big) T
 - h^p \ell(\alpha,-\mu),
 & \hbox{if $(\alpha,h^p)\in\widetilde{\UUU}_X$,} \\
  -\infty, & \hbox{otherwise.} \\
\end{array}%
\right.
\end{align}
Moreover, if conditions \ref{cond:EM} and \ref{cond:COE} are satisfied, this
function is concave, upper semicontinuous and coercive; cf. Corollary
\ref{corr:spec_fun_properties}.

\begin{remark}
Let us briefly discuss for which \lev processes conditions \ref{cond:EM} and
\ref{cond:COE} are satisfied.  Condition \ref{cond:EM} is standard in
mathematical finance and is satisfied by the majority of \lev models, for
example, by the generalized hyperbolic, the CGMY and the Meixner processes. The
set $\UUU_Z$ is bounded for the majority of \lev models, in particular for the
aforementioned ones. The only exceptions popular in mathematical finance are
Brownian motion and Merton's jump-diffusion model. In these cases however the
existence of a Brownian part ensures that \ref{cond:COE} is satisfied.
\end{remark}

\subsection{The investors' optimization problem revisited}
\label{subs:spec_revisited}

The investors in this setting can trade continuously in the asset $S$ which
incorporates the financial market according to an admissible strategy $\theta$.
In other words, the set of trading outcomes equals
\begin{align*}
\mathcal{G}
  = \bigg\{ G_T(\theta) = \int_0^T \theta_u\ud S_u : \theta\in\Theta
    \bigg\},
\end{align*}
where the set of admissible trading strategies is defined by
\begin{equation}\label{def:Theta}
\Theta
  = \big\{ \theta \in L(S): G(\theta)
      \text{ is a } \QQ\text{-martingale for every } \QQ\in\MM_f \big\},
\end{equation}
while $L(S)$ denotes the set of predictable, $S$-integrable processes and
$\MM_f$ the set of absolutely continuous local martingale measures with finite
entropy, that is
\begin{equation}\label{eq:MMf}
\MM_f
  = \big\{ \QQ\ll\P \text{ on } \F_T:  S \text{ is a }
      \QQ\text{-local martingale and } \mathcal{H}(\QQ|\PP)<\infty \big\}.
\end{equation}
The \ref{cond:NA} condition is subsequently adjusted to the following one:
\begin{enumerate}[label={$(\mathbb{NA}')$},leftmargin=40pt]
\item $\MM_{f}\cap{\cal Q}_X\neq\emptyset$.\label{cond:NALevy}
\end{enumerate}
The investors' position \eqref{eq:investors_position} takes now the form
\begin{align}\label{eq:investors_position-Levy}
\overline{w}(\theta,h^s) &= h^s(P_T-F) + G_T(\theta),
\end{align}
and the aim is to derive an explicit expression for their optimization problem,
in particular for the function $u_s(\alpha,h^s)$ in \eqref{eq:def-us}.

Define the measure $\P_\mathpzc{s}$ via the Radon--Nikodym derivative
\begin{equation}\label{eq:measure_s}
\frac{\ud \P_{\sss}}{\ud \P}
 = \frac{\exp\left(-\gamma_s h^sP_T\right)}
    {\E\sbrac{\exp\left(-\gamma_s h^sP_T\right)}}
 \stackrel{\eqref{eq:fin_price}}{=}
   \frac{\exp\left(-\frac{\gamma_s h^s}{m}X\right)}
        {\EE\sbrac{\exp\left(-\frac{\gamma_sh^s}{m}X\right)}}
 \stackrel{\eqref{eq:model-prod}}{=}
    \frac{\exp\left(-\scal{\frac{\gamma_s h^s}{m}u_2}{Z_T}\right)}
        {\EE\sbrac{\exp\left(-\scal{\frac{\gamma_s h^s}{m}u_2}{Z_T}\right)}},
\end{equation}
for every $h^s$ such that $-\frac{\gamma_s h^s}{m}u_2\in\UUU_Z$. The following 
lemma provides the dynamics of the process $Z$ under $\P_\sss$.

\begin{lemma}\label{lem:triplet}
The process $Z$ remains a \lev process under $\P_\sss$ with cumulant generating
function provided by
\begin{align}
\kappa^{\mathpzc{s}}(v)
 = \kappa\left(v+\xi\right) - \kappa\left(\xi\right),
\end{align}
where $\xi:=-\frac{\gamma_sh^s}{m}u_2$, for all $v\in\R^d$ such that
$v+\xi\in\UUU_Z$. Moreover, the \lev triplet of the univariate \lev process
$\scal{u_i}{Z}$, $i=1,2$, under $\P_\mathpzc{s}$ is provided by
\begin{eqnarray*}
b_i^\sss &=& \scal{u_i}{b} + \scal{u_i}{c\xi}
     + \int\nolimits_{\R^d} \scal{u_i}{x}\big(\e^{\scal{\xi}{x}}-1\big)
    \nu(\dx)\\
c_i^\sss &=& \scal{u_i}{cu_i} \\
\nu_i^\sss(E) &=&
        \int\nolimits_{\R^d} 1_E(\scal{u_i}{x}) \, \e^{\scal{\xi}{x}}\nu(dx),
\qquad E\in\mathcal{B}(\R^d).
\end{eqnarray*}
\end{lemma}

\begin{proof}
See e.g. \citet[Theorem VII.3.1]{Shiryaev99} for the first part and
\citet[Theorem 4.1]{EberleinPapapantoleonShiryaev08} for the second.
\end{proof}

The exponential transform of the process $Y=\scal{u_1}{Z}$ is denoted by
$\widetilde{Y}$, that is $\EN(\widetilde{Y})=\e^{Y}$. The process
$\widetilde{Y}$ is again a \lev process and its triplet, relative to $\P_\sss$,
is given by
\begin{align}\label{eq:exp-trans-triplet}
\widetilde{b}^\sss_1
 &= b^\sss_1+\frac{c^\sss_1 }{2} + \int\nolimits_\R(\e^x-1-x)\nu^\sss_1(\dx)
  = \kappa_1^\sss(1) \nonumber\\
\widetilde{c}^\sss_1
 &= c^\sss_1 = c_1 \\ \nonumber
\widetilde{\nu}^\sss_1(E)
 &= \int\nolimits_\R 1_E(\e^x-1) \nu^\sss_1(\dx), \quad E\in\mathcal{B}(\R);
\end{align}
see \citet[Lemma 2.7]{KallsenShiryaev02}. Here, $\kappa_1^\sss$ denotes the
cumulant generating function of $Y$ under $\P_\sss$ and is given by
\eqref{eq:Levy-cumulant} using the triplet $(b^\sss_1,c^\sss_1,\nu^\sss_1)$.

Now, recalling \eqref{eq:investors_utility} and
\eqref{eq:investors_position-Levy}, the investors' utility takes the
following form:
\begin{align}\label{specs_ut_1}
\UU_s\big(\overline{w}(\theta,h^s)\big)
 &= -\frac{1}{\gamma_s} \log \E \Big[ \exp \Big(
      -\gamma_s \big[ h^s(P_T-F) + G_T(\theta) \big] \Big) \Big] \nonumber\\
 &\stackrel{\eqref{eq:fin_price}}{=}
    -\frac{1}{\gamma_s} \log \E \Big[ \exp \Big(
      -\gamma_s \Big[ \frac{h^s}{m}X + G_T(\theta) \Big] \Big) \Big]
    + C_1(h^s,\alpha,F)  \nonumber\\
 &\stackrel{\eqref{eq:measure_s}}{=}
    -\frac{1}{\gamma_s} \log \E_\sss
      \Big[ \exp \Big( -\gamma_s G_T(\theta) \Big) \Big]
    + C_1(h^s,\alpha,F) - C_2(h^s),
\end{align}
where
\begin{align}\label{specs_ut_2}
C_1(h^s,\alpha,F)
  := h^s\left(\phi_0(\pi_T)-\frac{\alpha(1-\varepsilon)}{m}-F\right)
 \quad \text{ and } \quad
C_2(h^s) := \frac{T}{\gamma_s}\kappa_2\left(-\frac{\gamma_sh^s}{m}\right).
\end{align}

The next result provides the solution of the optimization problem with respect
to the financial market. We will also make use of the following condition:
\begin{enumerate}[label={$(\mathbb{FE})$},leftmargin=40pt]
\item \textit{There exists $\eta_*\in\R$ such that} \label{cond:FE}
      \begin{align}
    \int_{\{x>1\}} \e^x\e^{\eta_*\e^x}\nu_1^\sss(\dx) < \infty,
      \end{align}
      \textit{which solves the equation}
      \begin{align}\label{def:theta*}
    \frac{\partial}{\partial v}
    \widetilde{\kappa}^\sss_1(v)\big|_{v=\eta_*} = 0,
      \end{align}
      \textit{where $\widetilde{\kappa}^\sss_1$ denotes the cumulant generating
      function of $\widetilde{Y}$ under $\P_\sss$.}
\end{enumerate}

\begin{proposition}\label{specs_fin_market}
Assume that \ref{cond:EM} and \ref{cond:FE} hold. Then
\begin{align}
\sup_{\theta\in\Theta} \left\{
 -\frac{1}{\gamma_s} \log \E_\sss
      \Big[ \exp \Big( -\gamma_s G_T(\theta) \Big) \Big] \right\}
 &= -\frac{1}{\gamma_s} \widetilde{\kappa}_1^\sss(\eta_*)T.
\end{align}
\end{proposition}

\begin{proof}
According to \citet[Theorem 4.2]{Fujiwara_2006} and using condition
\ref{cond:FE}, we have that
\begin{align}\label{eq:entropy}
\sup_{\theta\in\Theta} \left\{
 -\frac{1}{\gamma_s} \log \E_\sss
      \Big[ \exp \Big( -\gamma_s G_T(\theta) \Big) \Big] \right\}
 &= -\frac{1}{\gamma_s} \log \inf_{\theta\in\Theta} \E_\sss
      \Big[ \exp \Big( -\gamma_s G_T(\theta) \Big) \Big] \nonumber\\
 &= \frac{1}{\gamma_s} \inf_{\Q\in\mathcal{M}_f} \mathcal{H}(\Q|\P_\sss)
  = \frac{1}{\gamma_s} \mathcal{H}(\P_*|\P_\sss),
\end{align}
where $\P_*$ denotes the measure minimizing the relative entropy with respect
to $\P_\sss$.

The function $x\mapsto|\e^x-1|\e^{\eta_*(\e^x-1)}$ is submultiplicative and
bounded by $\e^x\e^{\eta_*\e^x}$ on $\{x>1\}$, thus condition \ref{cond:FE} in
conjunction with \citet[Theorem 25.3]{Sato99} and \eqref{eq:exp-trans-triplet}
yield that
\[
\E_\sss\Big[|\widetilde{Y}_T|\e^{\eta_*\widetilde{Y}_T}\Big] < \infty.
\]
Applying \citet[Theorems 4 and 8]{Hubalek_Sgarra_2006}, we get that the minimal
entropy martingale measure for $\e^{Y}$ exists and coincides with the Esscher
martingale measure for $\widetilde{Y}$. The latter is provided by
\begin{align}
\frac{\ud \P_*}{\ud \P_\sss}
 = \frac{\e^{\eta_* \widetilde{Y}_T}}
        {\E_\sss\big[\e^{\eta_* \widetilde{Y}_T}\big]},
\end{align}
where $\eta_*$ is the root of equation \eqref{def:theta*}. Finally, using the
martingale property of $\widetilde{Y}$ (cf.
\cite[Remark~4]{Hubalek_Sgarra_2006}), we deduce that
\begin{align*}
\mathcal{H}(\P_*|\P_\sss)
 &= \E_*\big[ \eta_* \widetilde{Y}_T
    - \widetilde{\kappa}_1^\sss(\eta_*)T \big]
  = - \widetilde{\kappa}_1^\sss(\eta_*)T,
\end{align*}
which in turn implies the desired result.
\end{proof}

Therefore, using \eqref{specs_ut_1}--\eqref{specs_ut_2} and Proposition
\ref{specs_fin_market}, the investors' optimization problem can be written as
\begin{align}\label{eq:investors_problem_kappa_form}
\Pi^s
&= \sup_{\theta\in\Theta,\, h^s\in\R} \left\{ -\frac{1}{\gamma_s}
    \log \E_\sss \Big[ \exp \Big( -\gamma_s G_T(\theta) \Big) \Big]
    + C_1(h^s,\alpha,F) - C_2(h^s) \right\} \nonumber \\
&= \sup_{h^s\in\R} \left\{ -\frac{T}{\gamma_s} \left(
        \widetilde{\kappa}_1^\mathpzc{s}(\eta_*)
       +\kappa_2\left(-\frac{\gamma_sh^s}{m}\right)\right)
    + h^s\left(\phi_0(\pi_T)-\frac{\alpha(1-\varepsilon)}{m}-F\right)\right\}.
\end{align}
In other words, recalling \eqref{eq:specs_prob-1} and \eqref{r2}, the
investors' optimization problem has the representation
\begin{equation}\label{eq:investors_problem_propo}
\Pi^s = \sup_{h^s\in\R} \big\{u_s(\alpha,h^s)-h^sF\big\}
      = -u_s^*(\alpha,F),
\end{equation}
where the function $u_s(\alpha,h^s)$ admits the explicit expression
\begin{align}\label{eq:us-Levy}
u_s(\alpha,h^s) = \left\{
\begin{array}{ll}
 - \frac{T}{\gamma_s} \left(
        \widetilde{\kappa}_1^\mathpzc{s}(\eta_*)
       +\kappa_2\left(-\frac{\gamma_sh^s}{m}\right)\right)
    + h^s\left(\phi_0(\pi_T)-\frac{\alpha(1-\varepsilon)}{m}\right),
& \hbox{if $-\frac{\gamma_sh^s}{m}u_2\in\UUU_Z$,} \\
 - \infty, & \hbox{otherwise.} \\
\end{array}%
\right.
\end{align}

\begin{remark}\label{rem:usc}
Using the upper semicontinuity and the smoothness of the cumulant generating 
function together with the inverse function theorem, it follows from the 
explicit expression \eqref{eq:us-Levy} that the function $h^s\mapsto 
u_s(\alpha,h^s)$ is upper semicontinuous. Thus, condition \ref{cond:COE} is 
automatically satisfied in the current setting (provided that \ref{cond:EM} and 
\ref{cond:FE} hold).
\end{remark}

\begin{remark}
Let us also discuss for which \lev \procs conditions \ref{cond:NALevy} and 
\ref{cond:FE} are satisfied. \ref{cond:NALevy} is rather mild since it requires 
the existence of an equivalent martingale measure (EMM) with finite entropy 
under which the random variable $X=\scal{u_2}{Z_T}$ has finite first moment. 
Explicit constructions of EMMs for \lev processes are studied in 
\citet{EberleinJacod97} and in \citet{Cherny_Shiryaev_2002}. \ref{cond:FE} is 
also standard in the literature related to exponential utility maximization and 
entropic hedging. \citet{Hubalek_Sgarra_2006} provide explicit parameter regimes 
for this condition to be satisfied, which fit well with empirical data.
\end{remark}

\begin{remark}\label{rem:NA condition}
Condition \ref{cond:NALevy} implies that the investors' indifference price for 
the commodity is bounded from above. More precisely, the (buyer's) indifference 
price for a random payoff $C_T$ is defined as the solution $\mathrm{p}(C_T)$ of 
the equation
$$\sup_{G\in\mathcal{G}}
\UU_s\big(G-\mathrm{p}(C_T)+C_T\big)=\sup_{G\in\mathcal{G}}\UU_s(G).$$
According to \citet[\S5.2]{Delbaen_etal_2002} or 
\citet[\S4]{Fujiwara_Miyahara_2003} (see also \citet{LaeSta14}), the 
indifference price of an agent with exponential utility and risk aversion equal 
to $\gamma_s$ admits the following representation
\begin{equation}\label{eq:duality_of_indifference}
\mathrm{p}(C_T)
 = \underset{\QQ\in\MM_f}{\inf} \Big\{
    \EE_{\QQ}[C_T]+\frac{1}{\gamma_s}\HH(\QQ|\PP) \Big\}
  -\frac{1}{\gamma_s}\HH(\QQ_*|\PP),
\end{equation}
where $\QQ_*$ is the martingale measure minimizing the entropy with respect to 
$\PP$. With this at hand, \eqref{eq:fin_price} yields the assertion.
\end{remark}

We conclude this subsection with a statement analogous to Proposition 
\ref{pro:producers_problem} for the investors' side, thereby strengthening the 
results of Proposition \ref{pro:investors_problem}. More specifically, we show 
that the investors' optimization problem admits a maximizer for every 
$\alpha\in[0,\pi_0]$ and every forward price in the no-arbitrage interval, which 
is defined by
\begin{align*}
\NA := \bigg(\underset{\QQ\in\MM_f}{\inf}
         \EE_{\QQ}[P_T],\underset{\QQ\in\MM_f}{\sup}\EE_{\QQ}[P_T]\bigg).
\end{align*}

\begin{proposition}\label{pro:speculator's_problem}
Assume that conditions \ref{cond:EM}, \ref{cond:FE} and \ref{cond:NALevy}
hold. Then, for every $F\in\NA$ and $\alpha\in[0,\pi_0]$ there exists a
maximizer $\hat{h}^s\in\R$ for the producers' problem $\Pi^s$ such that
$-\frac{\gamma_s}{m}\hat{h}^su_2\in\UUU_Z$.
\end{proposition}

\begin{proof}
By the definition of indifference valuation and the cash invariance property
of the utility functional $\UU_s$, we have that
\begin{equation}\label{eq: indifference}
u_s(\alpha,h^s) = \sup_{\theta\in\Theta} \UU_s \big(G_T(\theta)\big)
        + \mathrm{p}(h^s P_T).
\end{equation}
Building on the above representation, it suffices to show that
$\mathrm{p}(h^sP_T)-h^sF$ is concave, upper semicontinuous and coercive.
Concavity is readily implied by \eqref{eq:duality_of_indifference}, while upper
semicontinuity follows from the fact that $u_s$ is upper semicontinuous; cf.
Remark \ref{rem:usc}. As for coercivity, using again
\eqref{eq:duality_of_indifference} we get that for every $h^s>0$
\begin{eqnarray*}
  \mathrm{p}(h^s P_T) -h^sF &=& \underset{\QQ\in\MM_f}{\inf} \Big\{
    \EE_{\QQ}[h^sP_T]+\frac{1}{\gamma_s}\HH(\QQ|\PP) \Big\}
  -\frac{1}{\gamma_s}\HH(\QQ_*|\PP) -h^sF \\
   &=& h^s\left(\underset{\QQ\in\MM_f}{\inf} \Big\{
    \EE_{\QQ}[P_T]+\frac{1}{h^s\gamma_s}\HH(\QQ|\PP) \Big\}
  -\frac{1}{h^s\gamma_s}\HH(\QQ_*|\PP)-F\right).
\end{eqnarray*}
Moreover, it holds that
\[ \underset{\QQ\in\MM_f}{\inf} \Big\{
    \EE_{\QQ}[P_T]+\frac{1}{h^s\gamma_s}\HH(\QQ|\PP) \Big\}
  - \frac{1}{h^s\gamma_s}\HH(\QQ_*|\PP)
  \xrightarrow[h^s\rightarrow+\infty]{}
  \underset{\QQ\in\MM_f}{\inf}\{\EE_{\QQ}[P_T]\},\]
hence $\mathrm{p}(h^sP_T)-h^sF$ goes to $-\infty$ as $h^s\to+\infty$, for every
$F\in\NA$. The limit as $h^s\rightarrow-\infty$ follows by similar
argumentation and using the payoff $-P_T$ instead of $P_T$.
\end{proof}

\begin{remark}
Proposition \ref{pro:speculator's_problem} states that for every fixed pair of 
parameters $(\alpha,F)\in[0,\pi_0]\times\NA$, the individual problem of the 
investors admits a finite solution. This solution is unique if the indifference 
price $\mathrm{p}(h^s P_T)$ is strictly concave as a function of $h^s$. In view 
of representation \eqref{eq:duality_of_indifference}, strict concavity is 
guaranteed if \{$\E_{\QQ}[X]:\QQ\in\MM_f\}$ is not a singleton, meaning that the 
variate $X$ determining the consumers' demand is not a replicable payoff.
\end{remark}

\subsection{The equilibrium revisited}

Finally, we can further strengthen the result on the existence of an 
equilibrium in the current setting, by showing that the equilibrium forward 
price is unique and belongs to the no-arbitrage interval. 

\begin{proposition}\label{thm:levy}
Assume that conditions \ref{cond:EM}, \ref{cond:COE}, \ref{cond:USC}, 
\ref{cond:NALevy} and \eqref{cond:boundedF} hold. Then there exists an
equilibrium $(\hat\alpha,\hat{h},\hat{F})$, where $\hat{F}\in\NA$ is unique.
\end{proposition}

\begin{proof}
In view of Theorem \ref{thm:main}, we only need to show that $\hat{F}\in\NA$ 
and is unique. Assume, for instance, that $\hat{F}\leq 
\underset{\QQ\in\MM_f}{\inf}\EE_{\QQ}[P_T]$. Taking into account the proof of 
Theorem \ref{thm:main} as well as representations 
\eqref{eq:duality_of_indifference} and \eqref{eq: indifference}, we get that
\begin{eqnarray*}
\hat{h}^s
 &=& \underset{h^s\in\mathbb{R}}{\mathop{\rm argmax}}
      \left\{ u_s(h^s,\hat{\alpha})-h^s \hat{F}\right\}
  = \underset{h^s\in\mathbb{R}}{\mathop{\rm argmax}}
      \left\{ \mathrm{p}(h^s P_T(\hat{\alpha}))-h^s\hat{F}\right\} \\
 &=& \underset{h^s\in\mathbb{R}}{\mathop{\rm argmax}}
      \left\{\underset{\QQ\in\MM_f}{\inf}
    \Big\{h^s\left(\EE_{\QQ}[P_T(\hat{\alpha})]
    -\hat{F}\right)+\frac{1}{\gamma_s} \HH(\QQ|\PP) \Big\} \right\}
  =+\infty
\end{eqnarray*}
The last statement contradicts the fact that $\hat{h}^s+\hat{h}^p=0$ and 
$\hat{h}^p\in\mathbb{R}$. The uniqueness of $\hat F$ follows from the smoothness 
of the cumulant generating function and the inverse function theorem, together 
with Remark \ref{rem:uniq}. 
\end{proof}

\begin{remark}
Assumption \ref{cond:FE} guarantees that there exists an optimal trading
strategy for the investors and is necessary in deriving the explicit
expression \eqref{eq:us-Levy}. However, it is not a necessary condition for the
existence of an equilibrium.
\end{remark}

\section{Examples, numerical illustrations and discussion}
\label{sec:illu}

In this final section, we consider two specific models for the evolution of the 
financial market and the consumers' demand. The first one is driven by 
correlated Brownian motions and the second one incorporates dependent jumps in 
addition. In the first case, we derive explicit expressions for the optimal 
storage policy and the optimal forward volume, and then the equilibrium price 
follows by the market clearing condition 
\eqref{eq:static_equilibrium_condition}. In the second case, we derive 
semi-explicit expressions for the optimal storage policy and the optimal forward 
volume, and the equilibrium price is then computed numerically. Thereafter, we 
study the effect of the various parameters, in particular the risk aversion 
coefficients of both agents and the production levels, in the formation of spot 
and forward prices.

\subsection{A model driven by Brownian motion}\label{subsec:bm model}

In the first example, the dynamics of the variates $X$ and $Y$ determining the
consumers demand and the financial market are driven by correlated Brownian
motions. Specifically
\begin{align}\label{eq:model-BM}
Y_t = b_1t + \sigma_1 W^1_t
 \quad \text{ and } \quad
X_t = \sigma_2 W_t^2,
\end{align}
where $W^1,W^2$ are standard Brownian motions with correlation $\rho\in[-1,1]$.
Moreover, using \eqref{eq:fin_price}, the mean and variance of the spot price
are given by
\begin{align}\label{eq: Exp}
\E[P_T] =  \phi_0(\pi_T) - \frac{\alpha(1-\varepsilon)}{m}
 \quad \text{ and } \quad
\Var[P_T] = \frac{\sigma_2^2 T}{m^2}.
\end{align}
The ensuing result provides an explicit expression for the optimal inventory
policy and the optimal investment in the forward contract.

\begin{proposition}\label{prop:BM}
Assuming the model dynamics provided by \eqref{eq:model-BM}, the optimal
strategy $(\hat\alpha,\hat{h}^p)$ for the producers' problem is given by
\begin{equation}\label{eq:producers_optimal_brown}
\hat\alpha
  = \left(\frac{\ud_3\ud_5-2\ud_2\ud_4}{4\ud_1\ud_4-\ud_3^2}\vee
    0\right)\wedge\pi_0
 \quad \text{ and } \quad
\hat{h}^p = -\frac{\hat{\alpha}\ud_3 + \ud_5}{2\ud_4},
\end{equation}
while the optimal position $\hat{h}^s$ for the investors' problem equals
\begin{align}
\hat{h}^s = \frac{\E[P_T]-F}{\bar{\gamma}_s\Var[P_T]}
      - \frac{\lambda\rho\sqrt{T}}{\bar{\gamma}_s\sqrt{\Var[P_T]}}.
\end{align}
Here, the constants $\ud_1,\dots,\ud_5$ are provided by
\eqref{eq:constants_prod_prob} and $\bar{\gamma}_s=\gamma_s(1-\rho^2)$.
\end{proposition}

The proof of the preceding Proposition is postponed for Appendix
\ref{app-proofs}.\smallskip

The equilibrium forward price $\hat{F}$ will be derived endogenously via the 
clearing condition \eqref{eq:static_equilibrium_condition}, where we should note 
that $\hat\alpha,\hat{h}^p$ and $\hat{h}^s$ all depend on $\hat{F}$. Thereafter, 
the equilibrium spot price of the commodity at the initial time is provided by
\begin{align}\label{eq:equil_spot_price}
P_0(\hat{F}) = \phi_0(\pi_0) + \frac{\hat{\alpha}(\hat{F})}{m}.
\end{align}
In this example, both the forward price and the optimal forward position are 
unique; this follows from Remark \ref{rem:uniq}, Proposition \ref{thm:levy} and 
the fact that $u_p(\hat\alpha,\cdot)$ and $u_s(\hat\alpha,\cdot)$ are strictly 
concave; see their explicit forms in \eqref{eq:producers_maximization_brown} and 
\eqref{eq:investors_maximization_brown}.

Figures \ref{fig:BM-1}, \ref{fig:BM-2} and \ref{fig:BM-3} exhibit how the 
storage amount, the forward volume, the spot price, the forward premium and the 
convenience yield at the equilibrium depend on the correlation between the 
consumers' demand and the financial market, as well as on the producers' and 
investors' risk aversion coefficients; see also the discussion in subsection 
\ref{discussion}.

\begin{remark}
Let us consider the case $\alpha^*=0$. Then, the optimal position for the
producers simplifies to
\begin{align}
\hat{h}^p(F) = \frac{\E[P_T]-F}{\gamma_p\Var[P_T]} - \pi_T
\end{align}
and the clearing condition \eqref{eq:static_equilibrium_condition} yields that
the equilibrium forward price is provided by
\begin{align}\label{eq:F_zero_alpha}
\hat{F} = \E[P_T]
 - \frac{\gamma_p\bar\gamma_s}{\gamma_p+\bar\gamma_s}\Var[P_T]
   \left( \frac{\lambda\rho\sqrt{T}}
      {\bar\gamma_s\sqrt{\Var[P_T]}} + \pi_T\right) .
\end{align}
\end{remark}

\begin{remark}\label{nf}
In case there does not exist a forward contract that the producers could use 
for hedging---hence, there are also no investors in the market---the producers' 
optimization problem takes the form
\begin{equation}\label{eq:producers_maximization_brown_NF}
\Pi^p_{\rm nf}
 = \underset{\alpha\in[0,\pi_0]}{\max}
   \big\{ \ud_1\alpha^2 + \ud_2\alpha + \ud_3^\prime \big\},
\end{equation}
where $\ud_1,\ud_2$ are given by \eqref{eq:constants_prod_prob}. Therefore,
the optimal storage strategy equals
\begin{align}
\hat{\alpha} = (\alpha^*\vee 0)\wedge\pi_0
    \quad \text{ with } \quad
\alpha^* = -\frac{\ud_2}{2\ud_1},
\end{align}
and the spot price of the commodity is
\begin{align*}
P_0(\hat{\alpha}) = \phi_0(\pi_0) + \frac{\hat{\alpha}}{m}.
\end{align*}
\end{remark}

\begin{figure}[h!]
 \centering
  \begin{minipage}{0.49\textwidth}
  \includegraphics[trim = 10mm 0mm 10mm 5mm, clip,width=7.6cm]{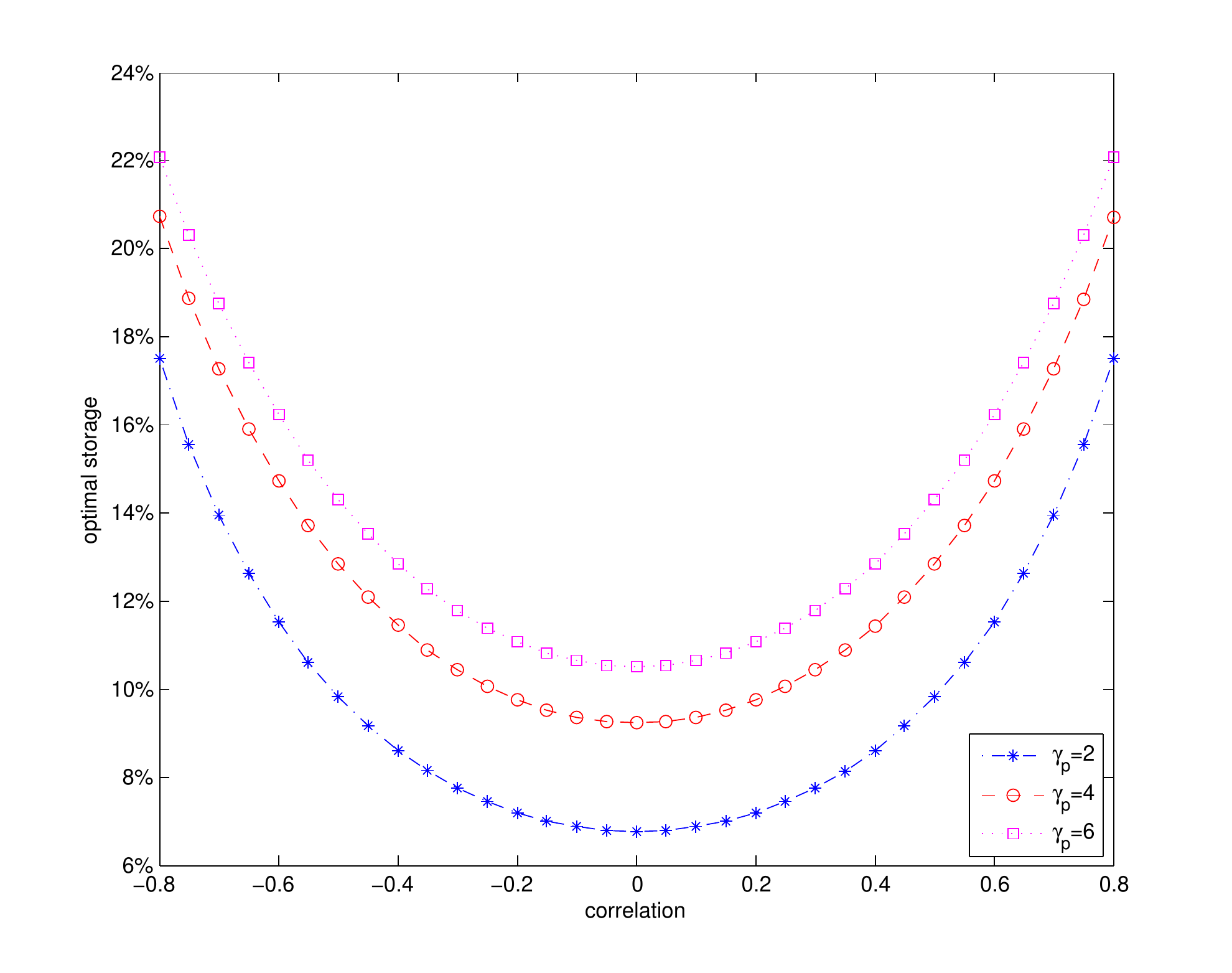}
  \end{minipage}
  \begin{minipage}{0.49\textwidth}
  \includegraphics[trim = 10mm 0mm 10mm 5mm, clip,width=7.6cm]{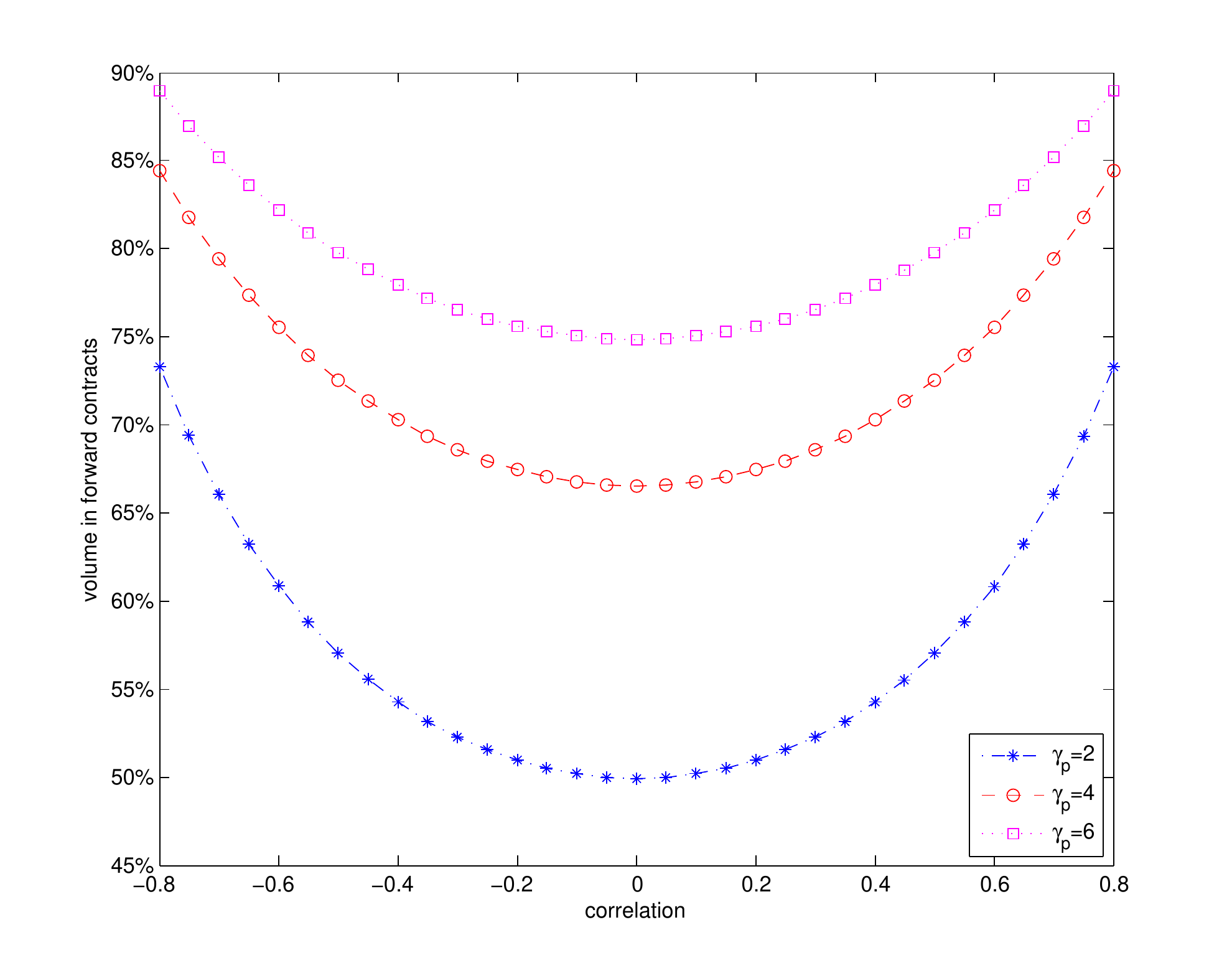}
  \end{minipage}
  \begin{minipage}{0.49\textwidth}
  \includegraphics[trim = 10mm 0mm 10mm 5mm, clip,width=7.6cm]{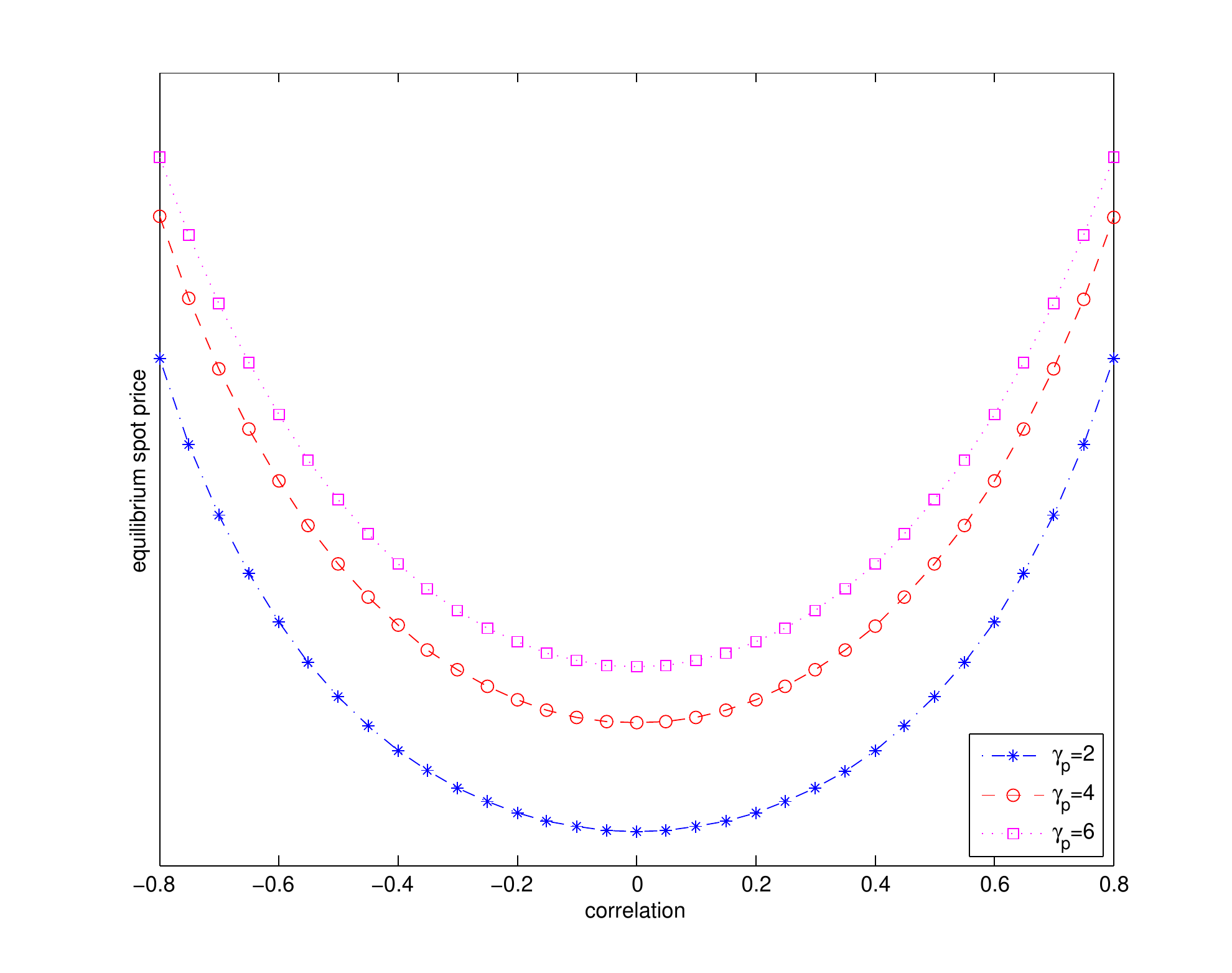}
  \end{minipage}
  \begin{minipage}{0.49\textwidth}
  \includegraphics[trim = 10mm 0mm 10mm 5mm, clip,width=7.6cm]{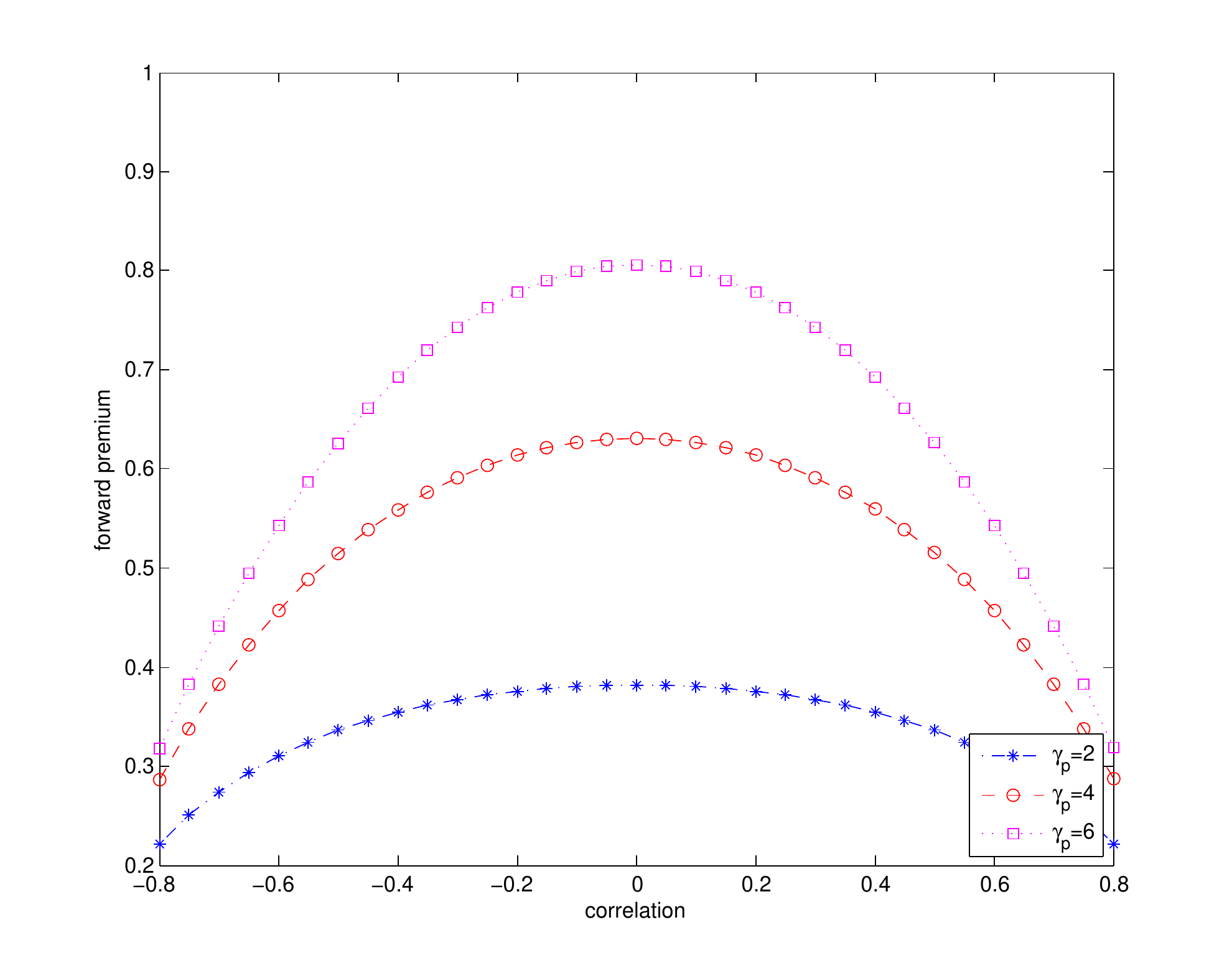}
  \end{minipage}
 \caption{Equilibrium storage amount (top-left), volume in forward contracts
      (top-right), spot price (bottom-left) and forward premium
      (bottom-right) as a function of correlation for different values of
      the producers' risk aversion $\gamma_p$.}
 \label{fig:BM-1}
\end{figure}

\begin{figure}[h!]
 \centering
  \begin{minipage}{0.49\textwidth}
  \includegraphics[trim = 10mm 0mm 10mm 5mm, clip,width=7.6cm]{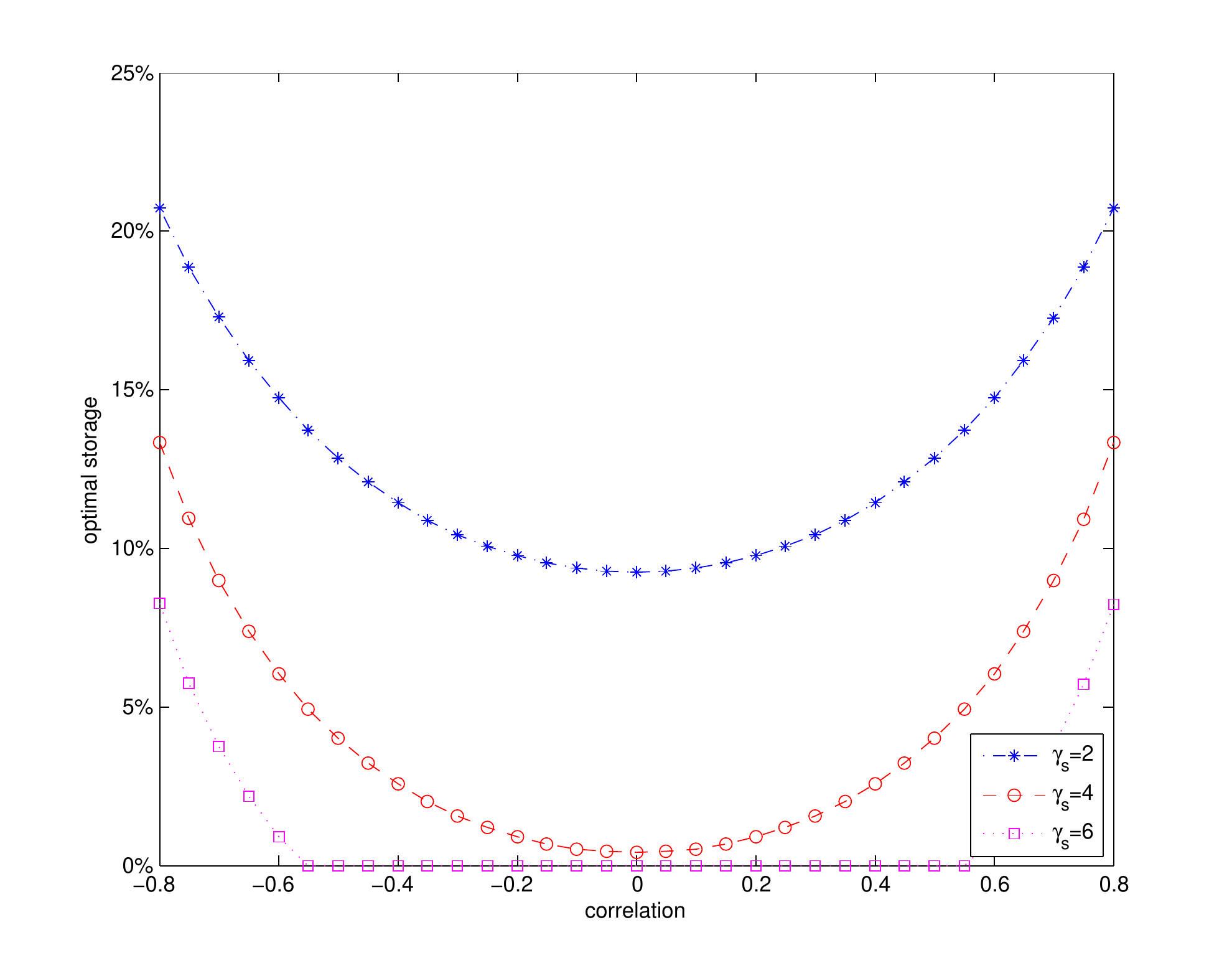}
  \end{minipage}
  \begin{minipage}{0.49\textwidth}
  \includegraphics[trim = 10mm 0mm 10mm 5mm, clip,width=7.6cm]{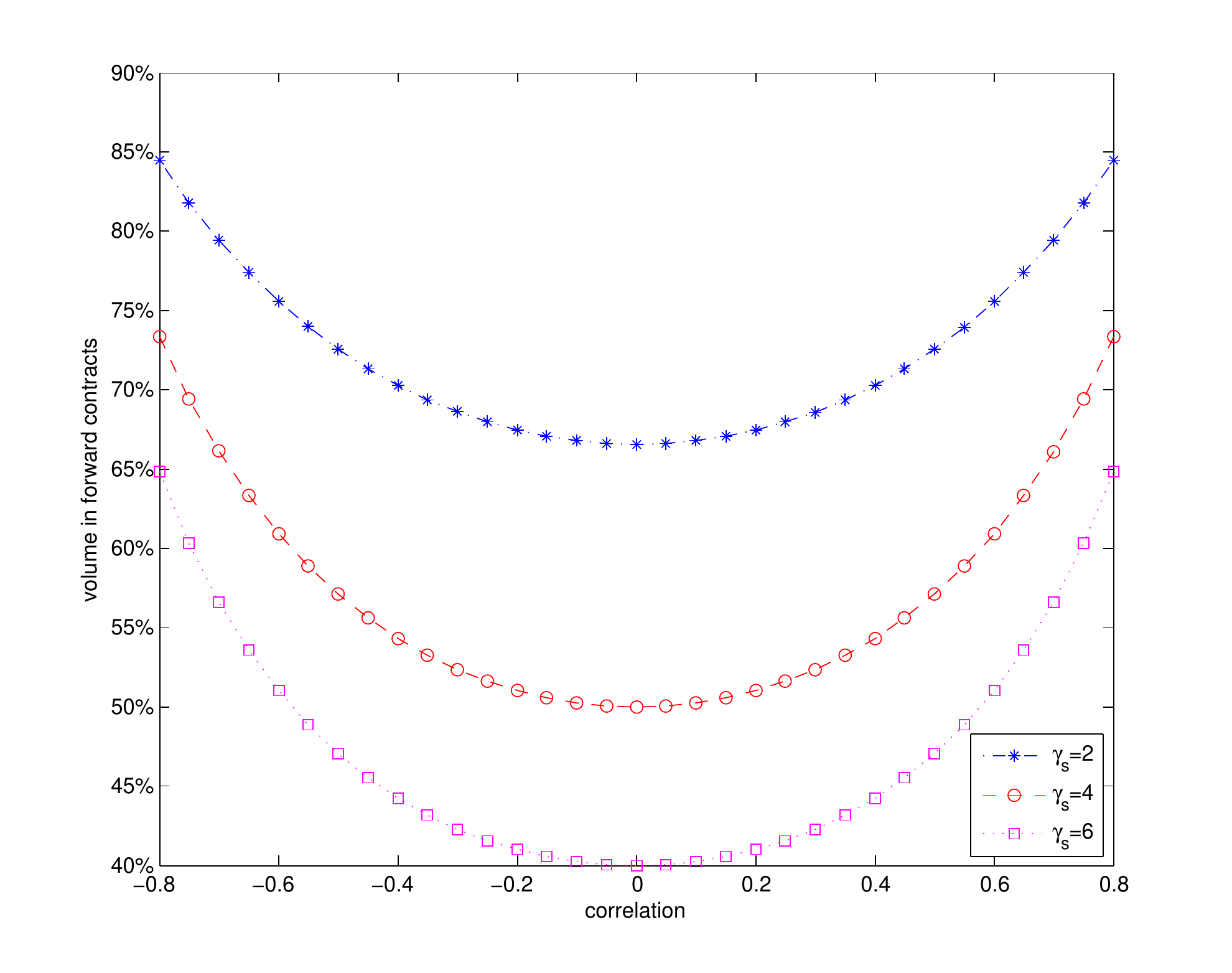}
  \end{minipage}
  \begin{minipage}{0.49\textwidth}
  \includegraphics[trim = 10mm 0mm 10mm 5mm, clip,width=7.6cm]{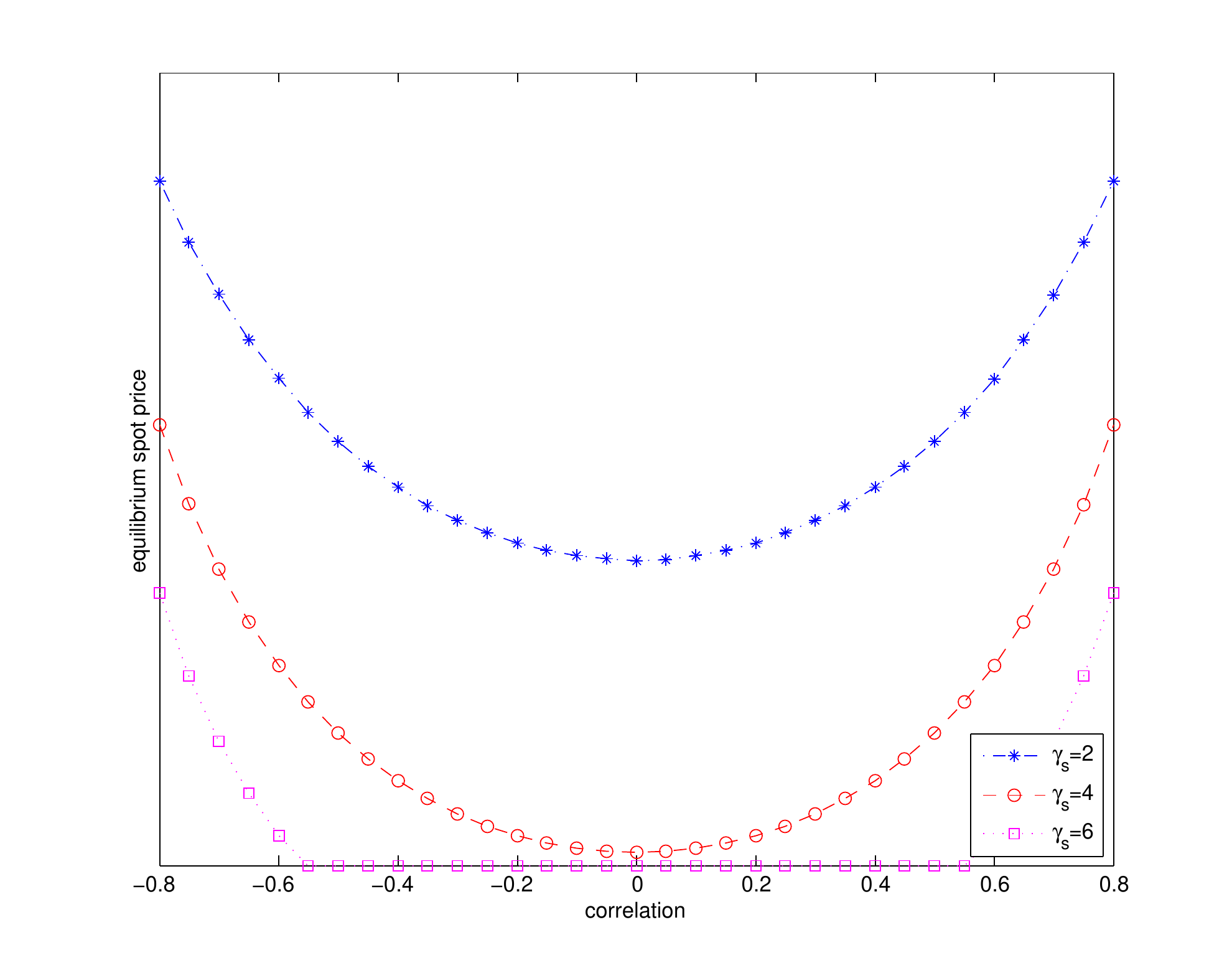}
  \end{minipage}
  \begin{minipage}{0.49\textwidth}
  \includegraphics[trim = 10mm 0mm 10mm 5mm, clip,width=7.6cm]{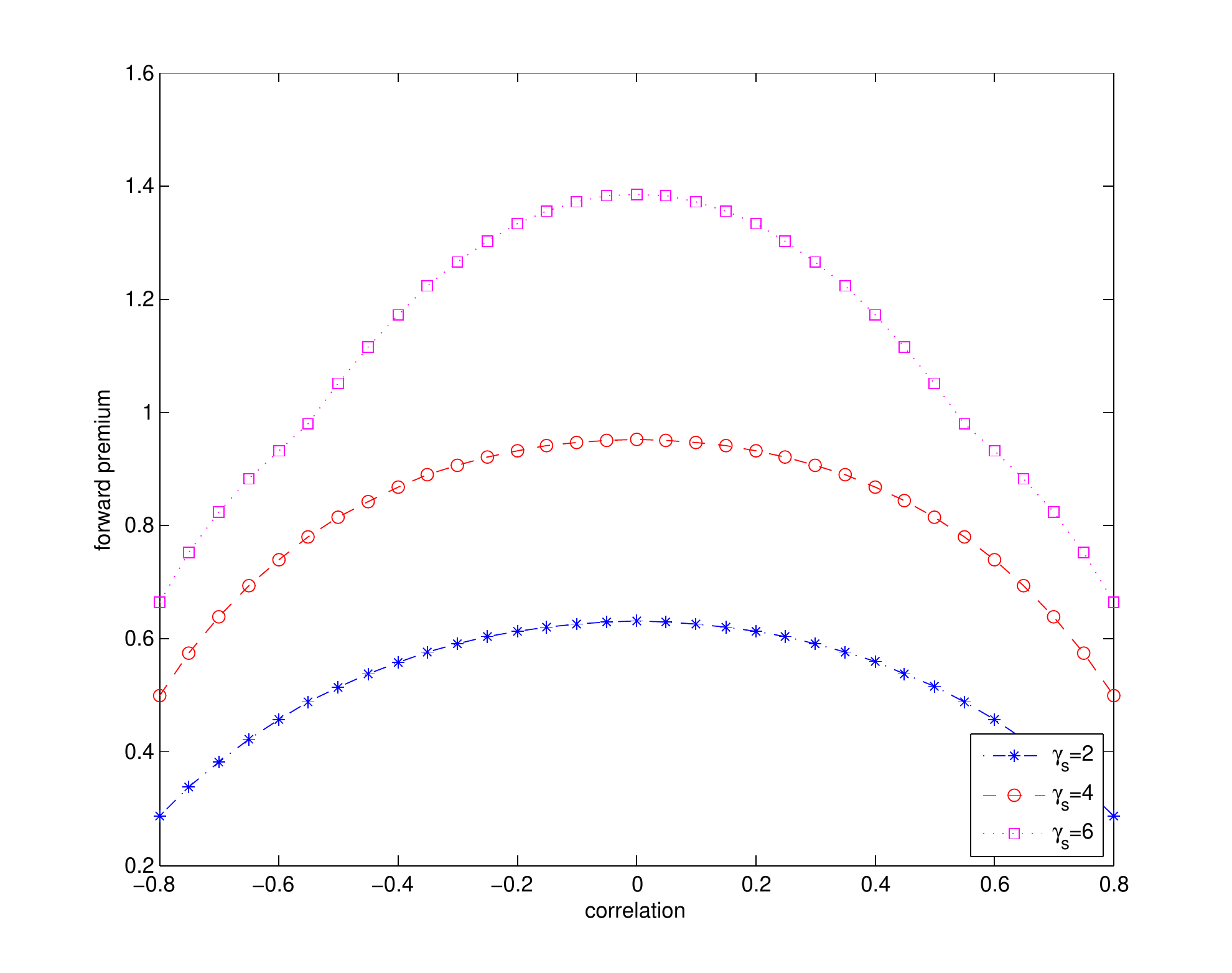}
  \end{minipage}
 \caption{Equilibrium storage amount (top-left), volume in forward contracts
      (top-right), spot price (bottom-left) and forward premium
      (bottom-right) as a function of correlation for different values of
      the investors' risk aversion $\gamma_s$.}
 \label{fig:BM-2}
\end{figure}

\begin{figure}[h!]
 \centering
  \begin{minipage}{0.49\textwidth}
  \includegraphics[trim = 10mm 0mm 10mm 5mm, clip,width=7.6cm]{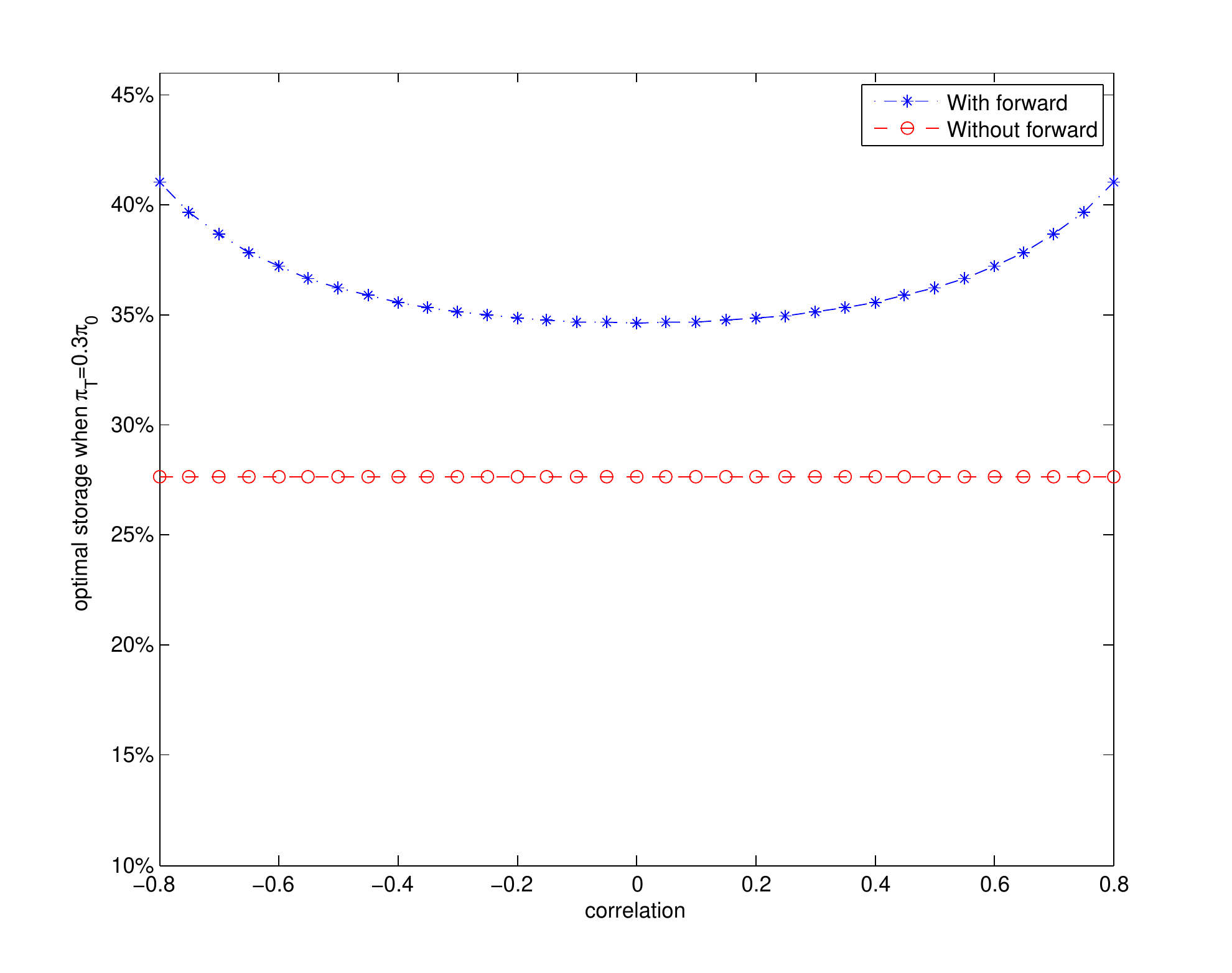}
  \end{minipage}
  \begin{minipage}{0.49\textwidth}
  \includegraphics[trim = 10mm 0mm 10mm 5mm, clip,width=7.6cm]{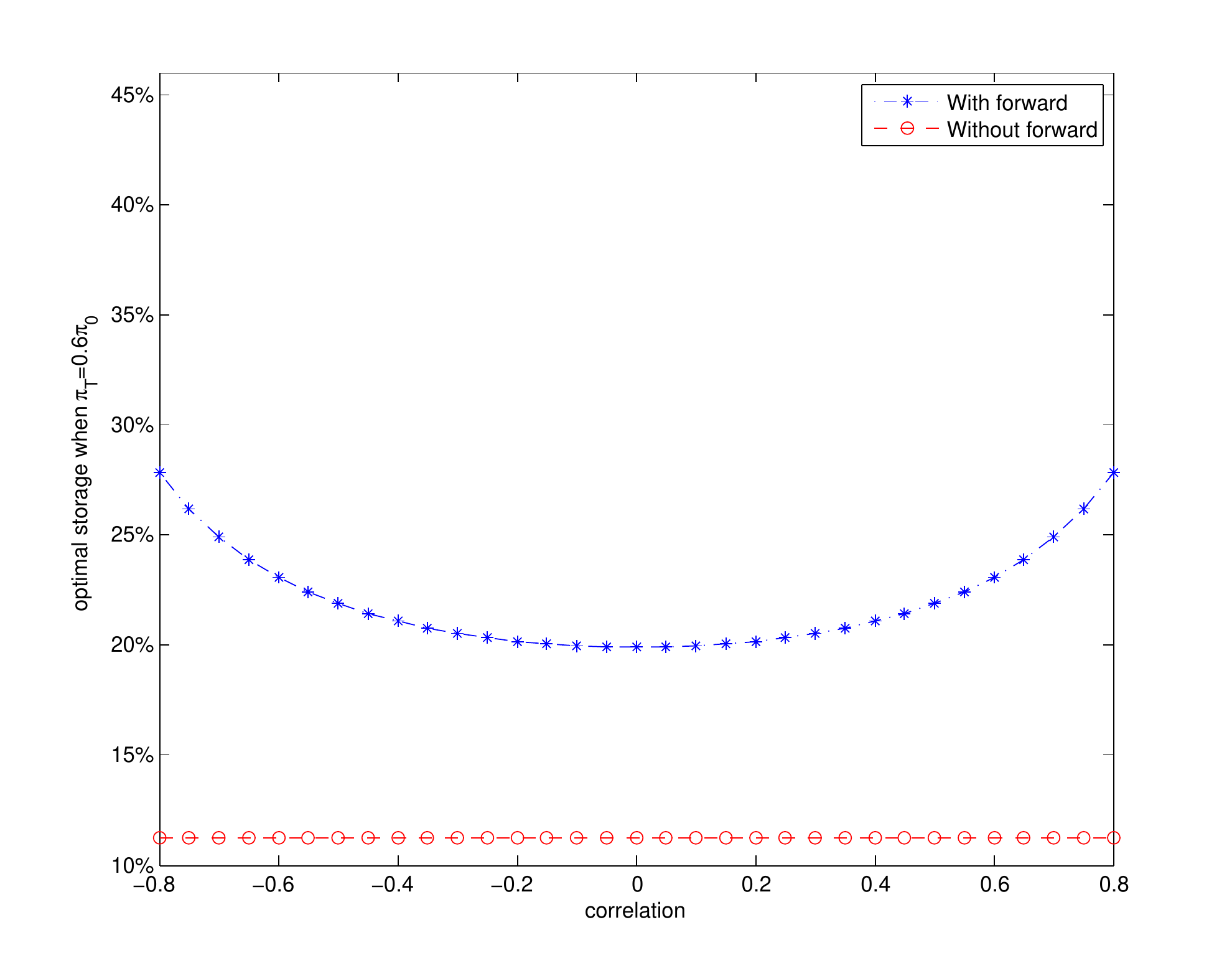}
  \end{minipage}
\caption{Equilibrium storage amount as a function of correlation for a market 
with and without forward contract. On the left $\pi_T=0.3\pi_0$ and on the right 
$\pi_T=0.6\pi_0$.}
 \label{fig:BM-4}
\end{figure}

\begin{figure}[h!]
 \centering
  \begin{minipage}{0.49\textwidth}
  \includegraphics[trim = 10mm 0mm 10mm 5mm, clip,width=7.6cm]{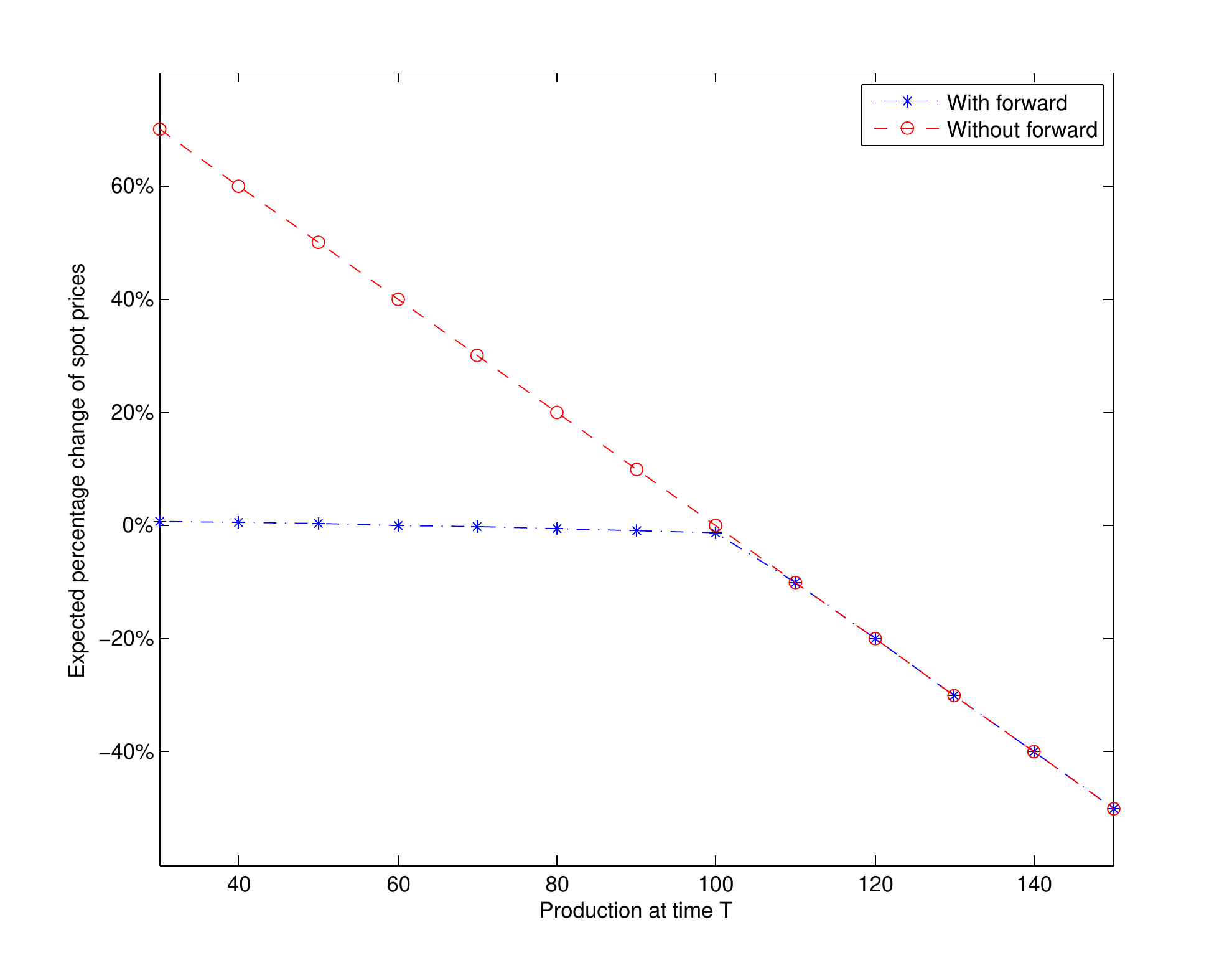}
  \end{minipage}
  \begin{minipage}{0.49\textwidth}
  \includegraphics[trim = 10mm 0mm 10mm 5mm, clip,width=7.6cm]{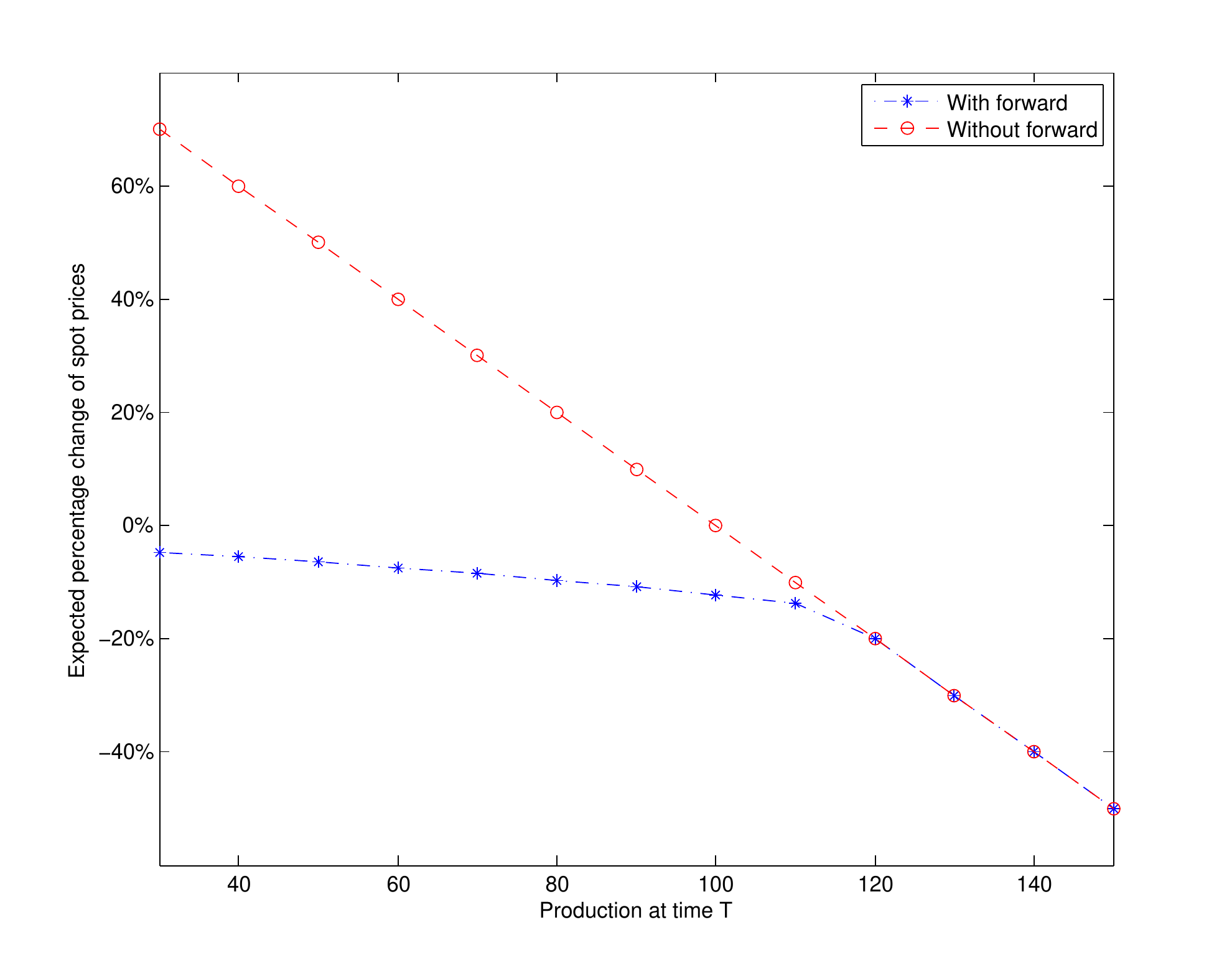}
  \end{minipage}
\caption{Expected percentage price changes 
$(\EE[\hat{P}_T]-\hat{P}_0)/\hat{P}_0$ as a function of the production $\pi_T$ 
(given that $\pi_0=100$) with and without forward contract. On the left, the 
correlation $\rho=0.2$ and on the right $\rho=0.7$.}
 \label{fig:BM-5}
\end{figure}

\subsection{A jump-diffusion model}

In the next example, the dynamics of the variates that determine the consumers'
demand and the financial market are driven by a \lev jump-diffusion process,
where the Brownian motion represents the `normal' market behavior while the
jumps appear simultaneously and represent some `shocks', e.g. news
announcements, that affect both the financial asset price and the demand for
the commodity. More precisely, the dynamics of the processes $Y$ and $X$ are
described by
\begin{align}\label{ljd-X}
Y_t = b_1t + \sigma_1 W_t^1 + \eta_1 N_t
 \quad \text{ and } \quad
X_t = b_2t + \sigma_2 W_t^2 + \eta_2 N_t,
\end{align}
where the drift term equals $b_i=\bar{b}_i-\lambda\eta_i$ with
$\bar{b}_i,\eta_i\in\R$ and $\sigma_i\in\R_+$, $i=1,2$. Furthermore, $W^1,W^2$
are standard Brownian motions with correlation $\rho$, while $N$ is a
univariate Poisson process with intensity $\lambda\in\R_+$. Hence, the constants
$\eta_1$ and $\eta_2$ represent the effect of a jump in the financial market and
the demand for the commodity, respectively.

Moreover, assuming $\bar{b}_2=0$ as in the previous example, the expectation of
$X_T$ equals zero and using \eqref{eq:fin_price} we get that
\begin{align}\label{ljd-exp}
\E[P_T] =  \phi_0(\pi_T) - \frac{\alpha(1-\varepsilon)}{m}
 \quad \text{ and } \quad
\Var[P_T] = \frac{\sigma_2^2 + \lambda\eta^2_2 }{m^2}T.
\end{align}
Observe that the presence of jumps, either negative or positive, increases the
variance of the spot price $P_T$ relative to the Brownian motion example. The
next result provides an expression for the optimal inventory policy and the
optimal investment in the forward contract.

\begin{proposition}\label{prop:LJD}
Assuming the model dynamics provided by \eqref{ljd-X}, the optimal strategy
$\hat\alpha,\hat{h}^p$ for the producers' problem is provided by $\hat\alpha =
(\alpha^*\vee0)\wedge\pi_0$ and $\hat{h}^p=h^{p,*}(\hat\alpha)$ where
$(\alpha^*,h^{p,*})$ solve the system of equations
\begin{align}\label{eq:jump_system_01}
\begin{cases}
 2\ud_1\alpha + \ud_2 + \ud_3h^p +
    \frac{\lambda\eta_2T(1-\varepsilon)}{m}\e^{-\gamma_p\eta_2\ell(\alpha,h^p)}
  &=0,\\
 \ud_3\alpha + 2\ud_4h^p + \ud_5 +
    \frac{\lambda\eta_2T}{m}\e^{-\gamma_p\eta_2\ell(\alpha,h^p)}
  &=0.
\end{cases}
\end{align}
Here $\ud_1,\dots,\ud_5$ are given by \eqref{eq:constants_prod_prob} by
replacing $\Var[P_T]$ with $\frac{\sigma_2^2T}{m^2}$. The optimal investment
for the investors' problem $\hat{h}^s$ is provided by the solution to the
equation 
\begin{align}\label{eq:jump_system_02}
\frac{\partial}{\partial h^s} \left\{ -\frac{T}{\gamma_s} \left[
    \widetilde{\kappa}^\sss_1(\eta_*)
  + \kappa_2\left(-\frac{\gamma_sh^s}{m}\right) \right] \right\} = F - \E[P_T],
\end{align}
where $\eta_*$ is given by \eqref{eq:jump_system_2}.
\end{proposition}

The proof of the preceding Proposition is postponed for Appendix
\ref{app-proofs}.\smallskip

Similarly to the previous example, the unique equilibrium forward price 
$\hat{F}$ is derived endogenously via the clearing condition
\eqref{eq:static_equilibrium_condition}, by noting again that 
$\hat\alpha,\hat{h}^p$ and $\hat{h}^s$ depend on $\hat{F}$, and the equilibrium 
spot price of the commodity at the initial time is again given by 
\eqref{eq:equil_spot_price}. Therefore, in order to determine the equilibrium we 
need to solve equations \eqref{eq:jump_system_01} and \eqref{eq:jump_system_02}. 
To this end, we have used numerical techniques, and have subsequently examined 
the impact of jumps on equilibrium quantities; see Figure \ref{fig:LJD-1} and 
the discussion in subsection \ref{discussion}.

\begin{remark}
Using an independent Brownian motion instead of the Poisson process in 
\eqref{ljd-X}, we can get the same first and second moments for $P_T$ as the 
ones in \eqref{ljd-exp}. This will also result in higher forward premia. 
However, jump processes are more appropriate models for the shocks that occur in 
random times and, in addition, jumps (in contrast to another Brownian motion) 
allow for asymmetries in the distributions, like fat tails and skewness. See 
also the discussion in the introduction of Section 4.
\end{remark}

\begin{figure}[h!]\label{fig:eta2}
 \centering
  \begin{minipage}{0.49\textwidth}
  \includegraphics[trim = 10mm 0mm 10mm 5mm, clip,width=7.6cm]{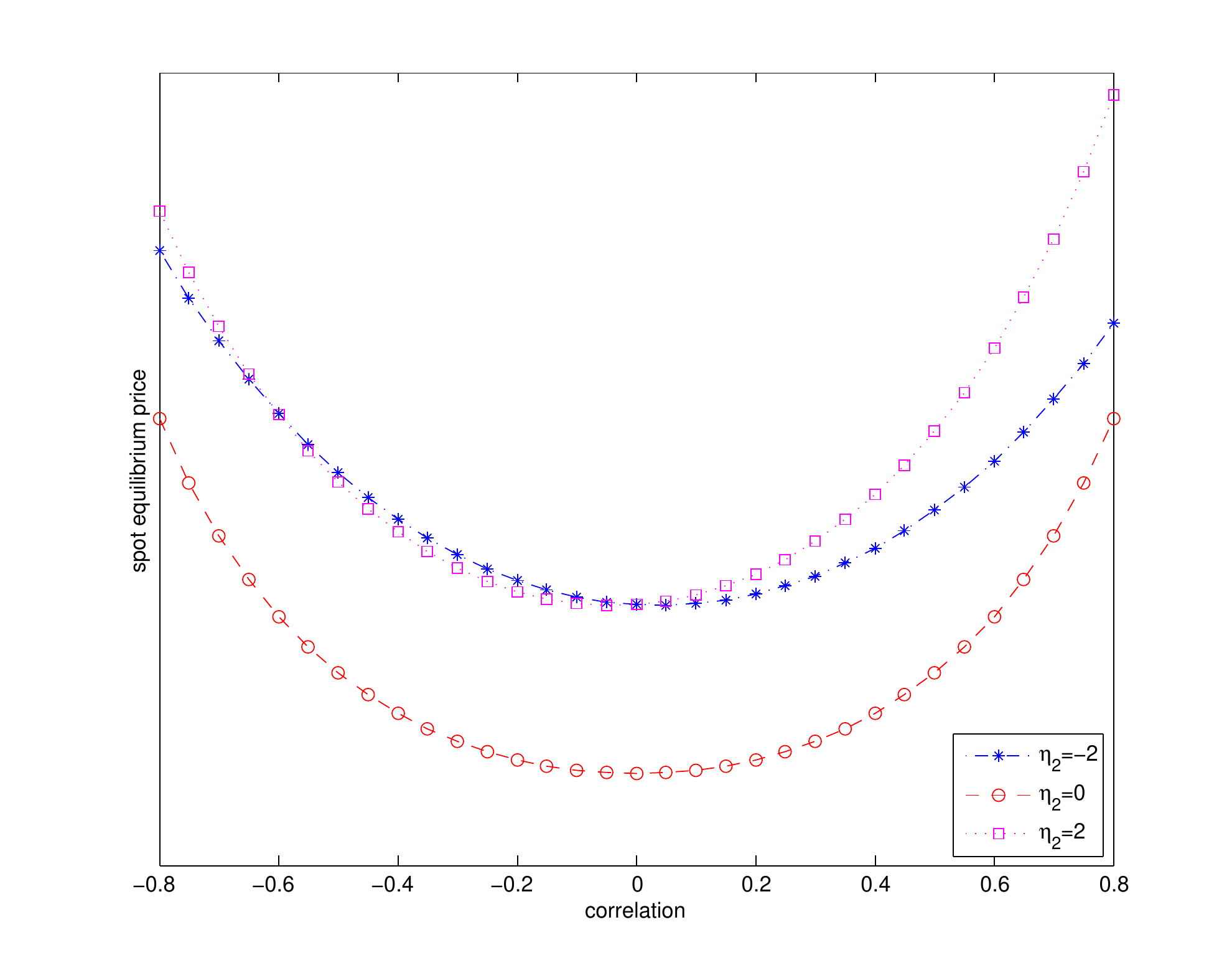}
  \end{minipage}
  \begin{minipage}{0.49\textwidth}
  \includegraphics[trim = 10mm 0mm 10mm 5mm, clip,width=7.6cm]{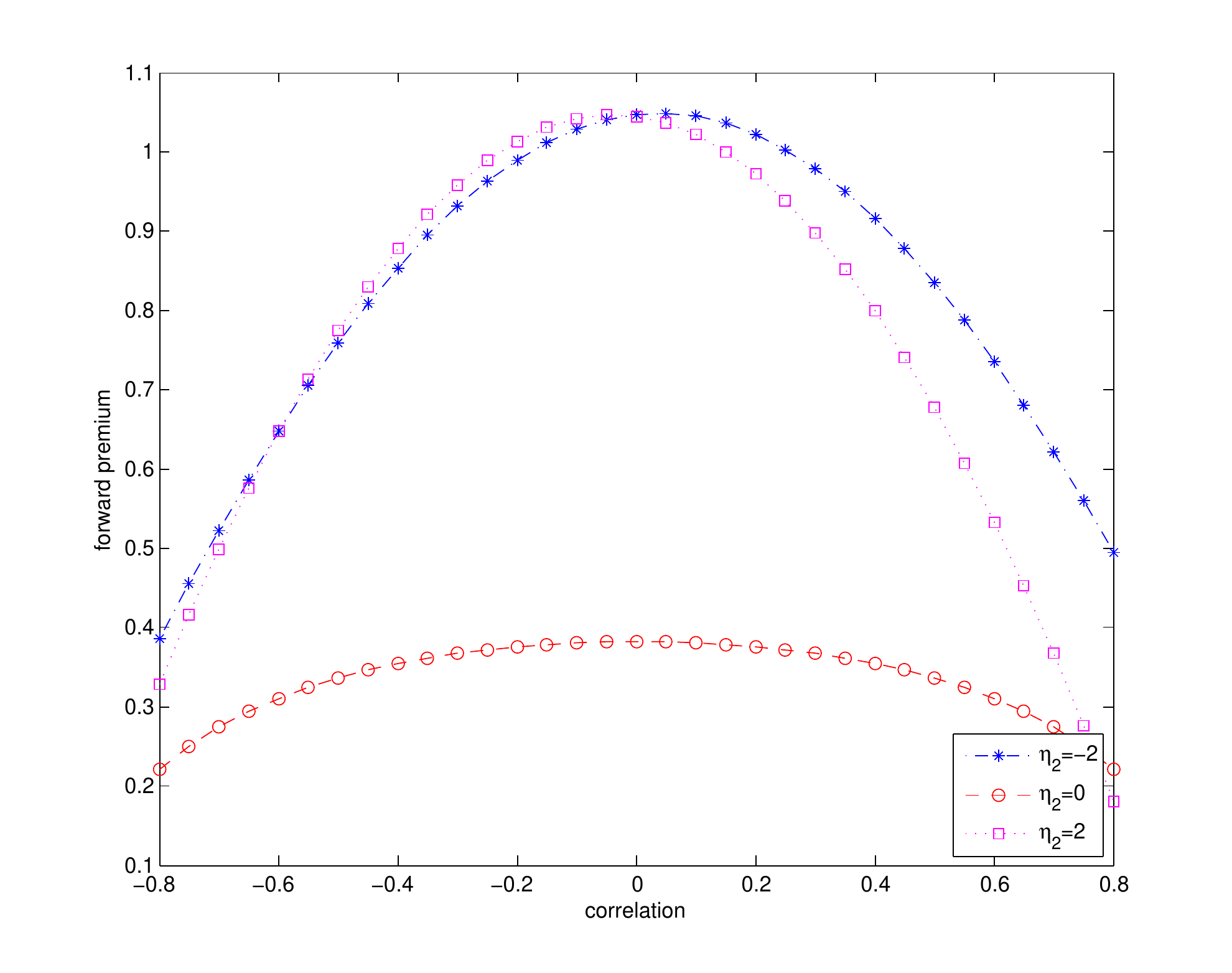}
  \end{minipage}
 \caption{Equilibrium spot price (left) and forward premium (right) as a
      function of correlation for different values of the demand shock
      effect $\eta_2$ (in this example $\eta_1=0$).}
 \label{fig:LJD-1}
\end{figure}

\subsection{Discussion of the results}
\label{discussion}

\subsubsection*{Producers' risk aversion and spot/forward prices}

We can use the results above to create several figures that illustrate the 
effect of the model parameters on the equilibrium quantities. We first examine 
the producers' side. The quantities that the producers have to consider are 
provided by $\underline{w}(\alpha,h^p)$ in \eqref{eq:producers_position}. We may 
split the terms into deterministic and stochastic ones. The deterministic part 
consists of the spot revenues from selling $\pi_0-\alpha$ units of the commodity 
at the spot price $P_0$, the expected future revenues from selling 
$\pi_T+\alpha(1-\varepsilon)$ units at the price $\EE[P_T]$, and the expected 
payoff of the short position $h^p$ in forward contracts. The stochastic term 
stems from the randomness of the future price $P_T$ and equals 
$[\alpha(1-\varepsilon)+h^p+\pi_T]X/m$. One can readily see that the 
deterministic term is decreasing with respect to $\alpha$. However, the risk in 
the stochastic term is also reduced for decreasing storage amounts. Assuming 
that $\E[X]=0$, this risk is minimized when the quantity 
$\alpha(1-\varepsilon)+h^p+\pi_T$ vanishes, that is, when all the future sales 
are hedged\footnote{Note that if the expectation of the random term $X$ is 
large, then producers are encouraged to increase the storage and supply more 
units at the terminal time. This speculative move explains how the storage may 
amplify the effects of a positive shock in demand which increases the spot 
price.}. Hence, a large amount of the commodity in storage implies also a large 
position to be hedged and vice versa (all else equal).

Considering only the deterministic term, and assuming that $\mu$ is 
sufficiently large, producers have motive to store their production only if 
$\pi_T$ is relatively smaller than $\pi_0$ (recall the discussion in Remark 
\ref{rem: storage}). In any case, storing part of their production now increases 
the spot price of the commodity. In addition, because producers are risk averse, 
in order to hedge their future risk exposure, they are willing to share some of 
their future revenues by taking a short position in the forward contract. 
Naturally, the higher the risk aversion the larger the short position in the 
forward contract (see the top-right of Figure \ref{fig:BM-1}) and the higher the 
forward premium paid to the investors (see the bottom-right of Figure 
\ref{fig:BM-1}). Moreover, a larger position in forward contracts implies an 
increasing tendency for storage, thus higher risk aversion leads to increased 
storage amounts (see the top-left of Figure\ref{fig:BM-1}). Summarizing, even 
when the production levels at time 0 and $T$ are close, producers with higher 
risk aversion tend to store more of their production when they can hedge the 
risk of future sales, a result which is consistent with the theory of storage, 
and this strategy increases the spot price of the commodity (see the bottom-left 
of Figure \ref{fig:BM-1}).

This result is further supported by the model without a forward contract in the 
market, see Remark \ref{nf}. There, we observe that the only motive for the 
producers to store the commodity stems from the possible uneven productions 
(i.e. the difference between $\pi_0$ and $\pi_T$). This motive to store is 
increased when partial hedging is possible through the trading in forward 
contracts. In fact, as illustrated in Figure \ref{fig:BM-4}, the optimal storage 
is always higher in the model with forward contract, for every level of uneven 
productions, while for $\pi_T$ close to or higher than $\pi_0$, the optimal 
storage without forward contract is zero. Thus, spot prices in the model without 
forward contract are always lower compared to the model with forward. However, 
higher storage implies that the future expected spot price decreases (see for 
instance relation \eqref{eq: Exp}), assuming that there is no rolling of the 
position in the forward contracts. Hence, while forward contracts tend to 
increase the spot commodity price, they also tend to decrease the future spot 
price. Therefore, the presence of forward contracts in the commodity market 
stabilizes prices when the production levels are uneven. This is apparent in 
Figure \ref{fig:BM-5}, where the expected price changes 
$(\EE[\hat{P}_T]-\hat{P}_0)/\hat{P}_0$ are illustrated for different values of 
$\pi_T$. In this example, we note that when there is scarcity of the commodity 
at time $T$, forward contracts serve to stabilize commodity spot prices. On the 
contrary, when the production at initial time is lower than that at terminal 
time, then the expected price difference remains the same with and without the 
forward contract.

Let us also discuss the effect of jumps in the equilibrium quantities. Figure 
\ref{fig:LJD-1} illustrates the effect of a possible side shock in the 
consumers' demand stemming from a jump. This jump not only increases the risk of 
the future price but it is also unhedgeable, since it is independent from the 
evolution of the stock market (we have assumed $\eta_1=0$). Therefore, the 
forward premium paid to the investors is higher, irrespective of the sign of the 
jump (see the right part of Figure \ref{fig:LJD-1}). Moreover, when the future 
price is riskier, recalling the discussion above, we conclude that the more risk 
averse the producers are the more they increase the amount they store and hence 
they also increase the spot price of the commodity (see the left part of Figure 
\ref{fig:LJD-1}). In addition, note that the sign of the jump makes little 
difference in the equilibrium quantities (if the expectation of the future 
demand shock is kept equal to zero).

The effect of the producers' risk aversion on market equilibrium can be used to
examine how the number of producers affects the equilibrium commodity prices.
In the present framework of CARA preferences, the parameter $1/\gamma_p$
measures the producers' \textit{aggregate} risk tolerance. Therefore, if the
number of producers increases, the parameter $\gamma_p$ decreases and the
analysis above implies that equilibrium spot prices are lower, as expected.

\subsubsection*{Investors' risk aversion and spot/forward prices}

Let us now examine the investors' side. When they become more risk averse, they 
are less willing to undertake the risk of a forward position. This is 
illustrated in Figure \ref{fig:BM-2} (top-right), where the percentage 
$\hat{h}/(\pi_T+\hat\alpha)$ (i.e. the percentage of forward contracts with 
respect to the total supply at time $T$) is plotted. Also, as the theory of 
normal backwardation states, more risk averse investors would require higher 
forward premium to enter into the forward contract. This premium is usually 
measured by the fraction $(\E[\hat{P}_T]-\hat{F})/\hat{F}$ which is plotted in 
Figure \ref{fig:BM-2} (bottom-right). On the other hand, a higher forward 
premium implies that hedging is more expensive for the producers, hence they 
intend to supply more in the spot market and store less; note that the 
optimal storage amount even equals zero in some cases as the top-left of Figure 
\ref{fig:BM-2} shows). Summarizing, when investors are more risk averse they 
invest less in forward contracts, which reduces the amount that producers can 
use for hedging; thus, producers offer more on the spot market, rendering 
equilibrium spot prices lower (see the bottom-left of Figure \ref{fig:BM-2}).

Turning our attention to the effect of the correlation between the consumers' 
demand and the financial markets' return, we note that the equilibrium 
quantities mainly depend on the square of $\rho$; this is basically because 
investors can go both long and short in the stock market. When $\rho^2$ 
increases, the effective risk aversion of the investors', which is 
$\bar{\gamma}_s=\gamma_s(1- \rho^2)$, decreases. Therefore, an increase of 
$\rho^2$ is eventually equivalent to a decrease of $\gamma_s$. This is expected 
because when the financial and the commodity markets are correlated, the 
investors can partially hedge the risk they undertake on a forward commodity 
contract by adjusting their investment strategy in the stock market accordingly. 
Hence, they become more risk tolerant. The dependence of the equilibrium 
quantities on the correlation coefficient $\rho$ is illustrated in Figure 
\ref{fig:BM-2}.

The effect of the investors' risk aversion on market equilibrium can be used to 
examine how the number of investors affects the equilibrium commodity prices. In 
the present framework of CARA preferences, the parameter $1/\gamma_s$ measures 
the investors' \textit{aggregate} risk tolerance. Hence, if the number of 
investors increases, the parameter $\gamma_s$ considered in the above analysis 
decreases. As we have seen, the latter implies, among other things, higher 
equilibrium spot prices. This theoretical result is consistent with the observed 
comovement of the amounts invested in the commodity forward contracts and the 
commodity spot prices (see the related discussion in the introduction).

\subsubsection*{Convenience yield, correlation and uneven productions}

As mentioned in the introduction, the convenience yield is a measure of the
implicit benefit that inventory holders receive. Positivity of the convenience
yield is consistent with the theory of storage. In our model, the convenience
yield denoted by $y$ solves the equation
\begin{equation}\label{eq:yield}
F=P_0\frac{1+R}{1-\varepsilon}-yP_0,
\end{equation}
see e.g. \cite{AchaLochRama13}. The relation of the yield with respect to the 
risk aversion coefficients of the producers and the investors is illustrated in 
Figure \ref{fig:BM-3}. As expected, $y$ is increasing with respect to both risk 
aversion coefficients (all else equal). The relation for the producers' side 
follows readily from Figure \ref{fig:BM-1}, since higher producers' risk 
aversion implies higher spot equilibrium price and higher forward premium (and 
also lower equilibrium forward price). Similarly, as the risk tolerance of the 
investors decreases, the cost of hedging increases, which makes producers sell 
more at the spot rather than storing and selling at a future date (see, in 
particular, the bottom-right of Figure \ref{fig:BM-2}).

The relation of the yield with respect to the correlation coefficient is more 
involved. When $\rho^2$ increases, there are two effects of opposite directions 
on the convenience yield. The first is negative and stems from the decrease of 
the investors' effective risk aversion, while the second is positive and comes 
from the corresponding increase of the spot price (see the bottom-left of Figure 
\ref{fig:BM-1}). The final outcome depends on the level of the risk aversions 
and the difference of production levels (see Figures \ref{fig:BM-3} 
and\ref{fig:BM-6}). In particular, assuming that production levels are close to 
each other, when producers are sufficiently risk averse (tolerant), $y$ is 
decreasing (increasing) in $\rho^2$. Note also that the steep increase of the 
convenience yield when $\rho$ approaches zero (right graph on Figure 
\ref{fig:BM-3}) occurs when the storage is zero (compare with the top-left of 
Figure \ref{fig:BM-2}), since in this case only the negative effect of $\rho^2$ 
in the convenience yield occurs (when the storage is zero, the spot price does 
not increase).

On the other hand, the difference between the production levels $\pi_0$ and 
$\pi_T$ could change the monotonicity of the convenience yield with respect to 
the correlation coefficient. Indeed, when production at time $T$ is sufficiently 
larger than the initial production, storage is getting lower and hence the 
negative effect of $\rho$ in the convenience yield prevails. As the difference 
$\pi_T-\pi_0$ decreases, the positive effect that stems from the increased spot 
price is getting more influential, especially when producers are less risk 
averse (right of Figure \ref{fig:BM-6}).  

Finally and as expected, for any correlation level, scarcity of commodity at 
initial time implies higher convenience yield (see both sides of Figure 
\ref{fig:BM-6}). In particular, when $\pi_T$ is sufficiently larger than 
$\pi_0$, the influence of the correlation $\rho$ on the convenience yield 
increases, a fact that reflects the producers' benefit from satisfying their 
increased hedging through the forward contract (the latter is more intense when 
producers are more risk averse).

\begin{figure}[h!]\label{fig:yield}
 \centering
  \begin{minipage}{0.49\textwidth}
 \includegraphics[trim = 10mm 0mm 10mm 5mm, clip,width=7.6cm]{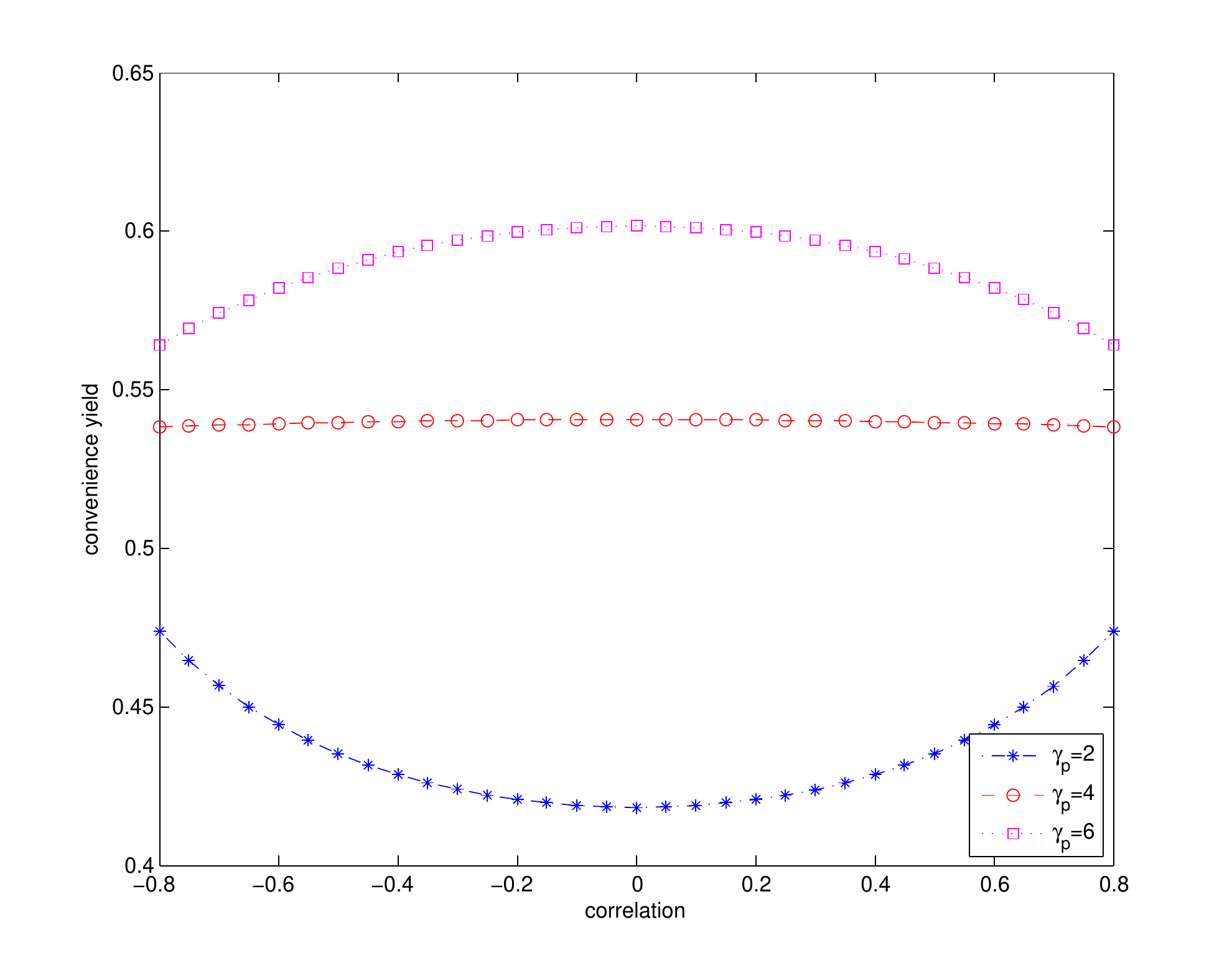}
\end{minipage}
 \begin{minipage}{0.49\textwidth}
 \includegraphics[trim = 10mm 0mm 10mm 5mm, clip,width=7.6cm]{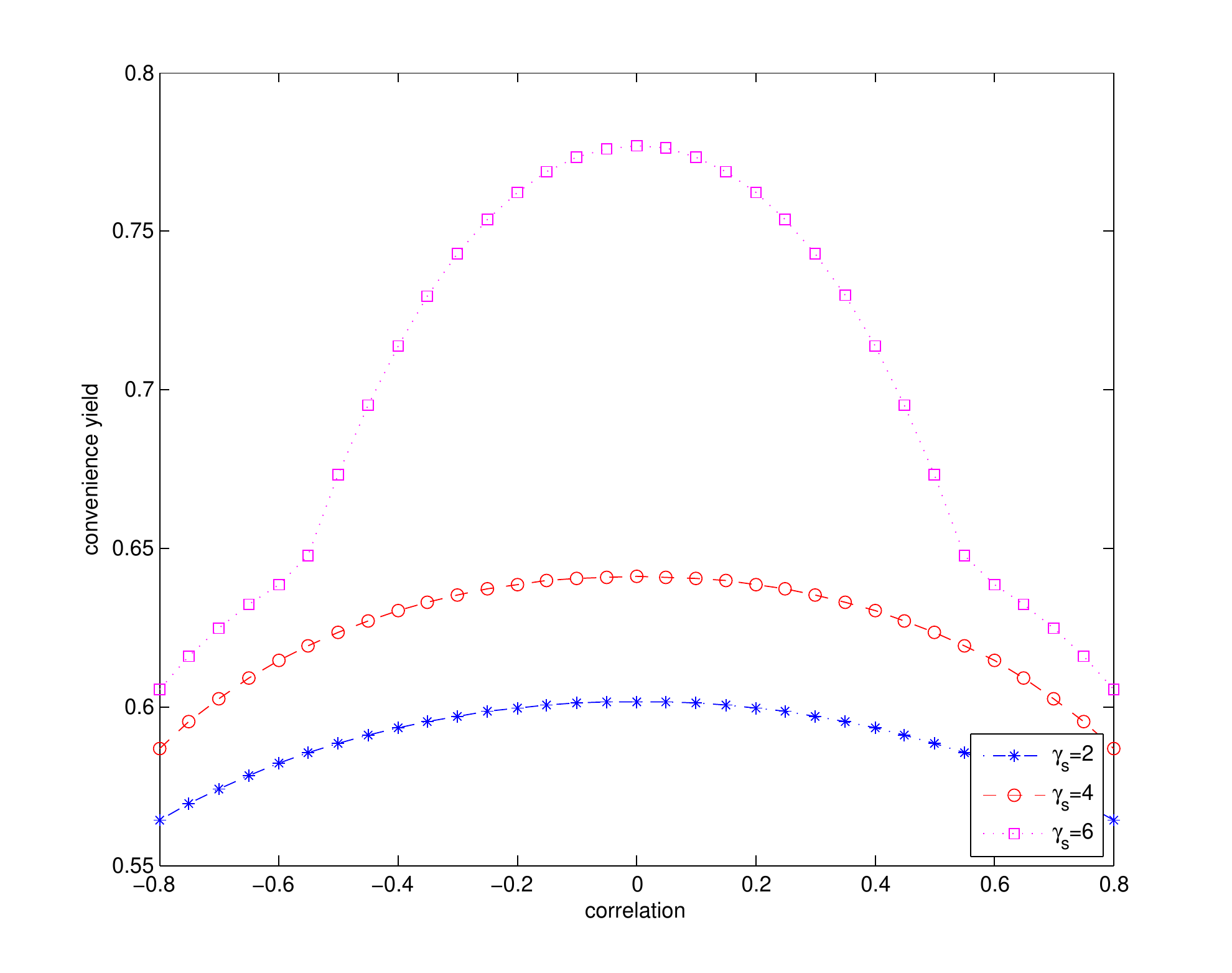}
\end{minipage}
\caption{Equilibrium convenience yield as a function of correlation for      
different values of producers' risk aversion $\gamma_p$ (left) and     
investors' risk aversion $\gamma_s$ (right), when the production levels are 
equal ($\pi_0=\pi_T$).}
 \label{fig:BM-3}
\end{figure}

\begin{figure}[h!]
 \centering
  \begin{minipage}{0.49\textwidth}
 \includegraphics[trim = 10mm 0mm 10mm 5mm, clip,width=7.6cm]{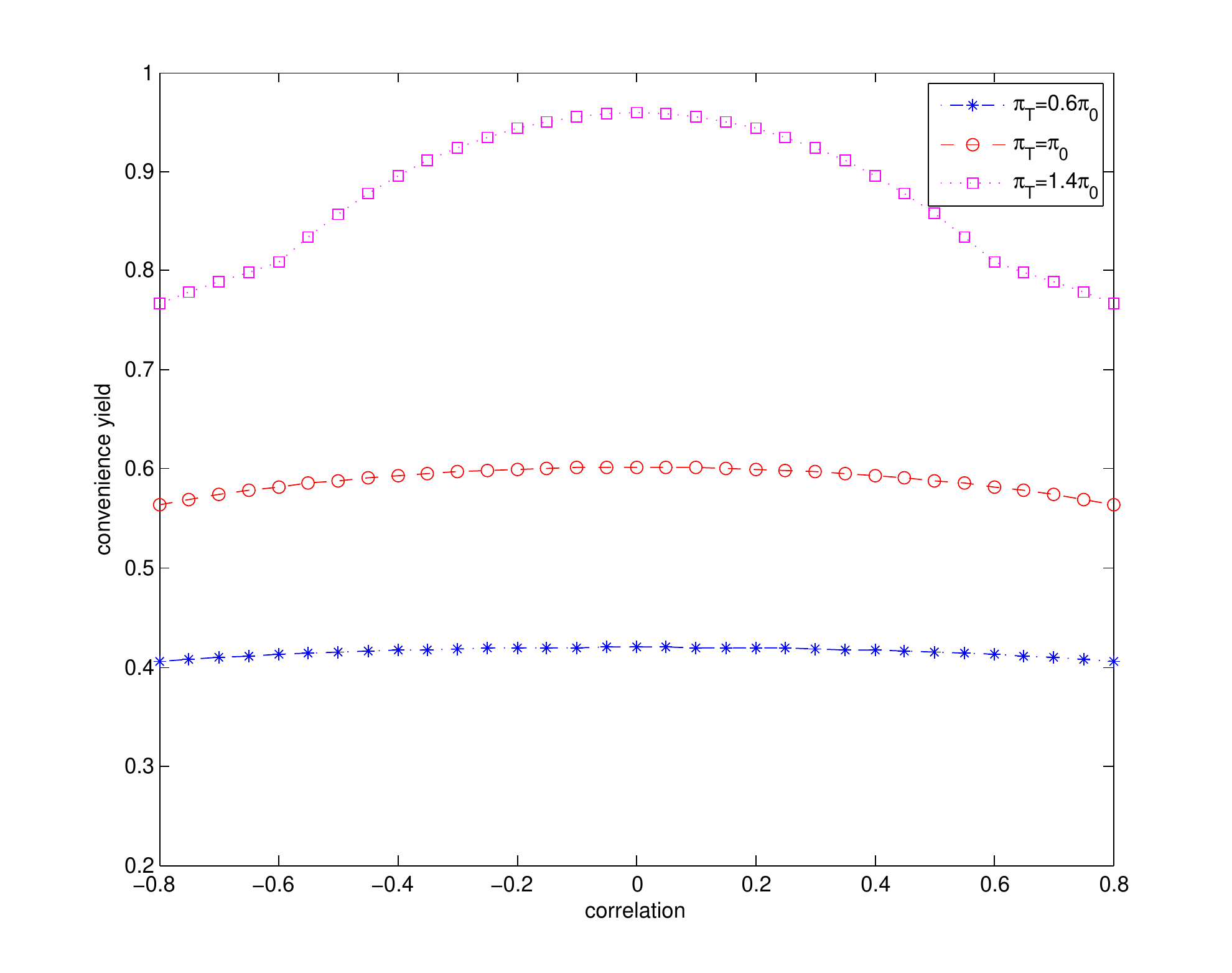}
\end{minipage}
 \begin{minipage}{0.49\textwidth}
 \includegraphics[trim = 10mm 0mm 10mm 5mm, clip,width=7.6cm]{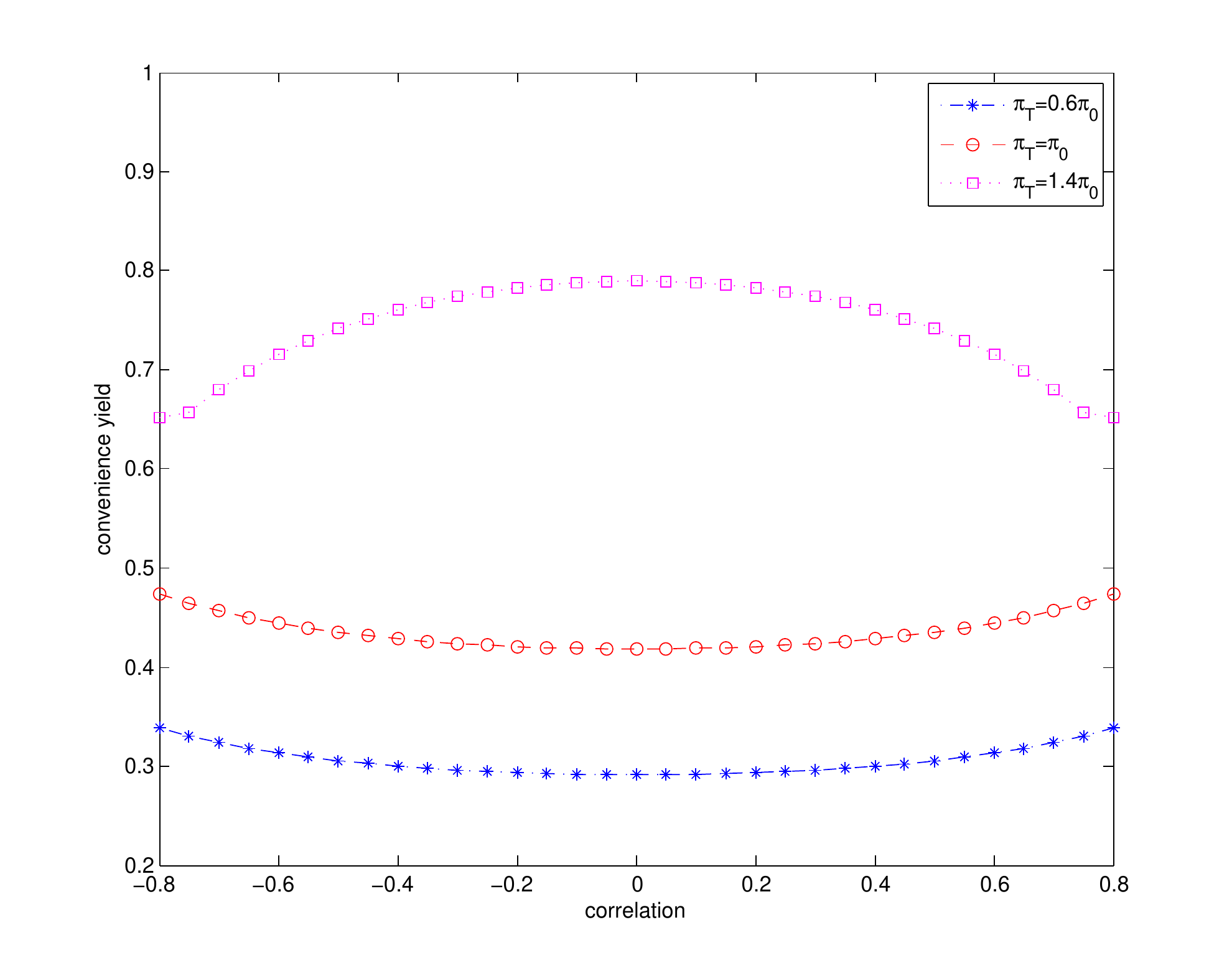}
\end{minipage}
\caption{Equilibrium convenience yield as a function of correlation for
     different values of production levels, when producers are more risk averse (left) and
     less risk averse (right).}
 \label{fig:BM-6}
\end{figure}

\begin{remark}
In \citet{AchaLochRama13}, the authors present an extensive empirical analysis, 
based on an equilibrium model simpler than the one we have established and 
developed above. In Section 3 therein, data from spot and future markets of oil 
and gas is used for testing the predictions of the model. Our model offers a 
much richer set-up, not only regarding the families of probability 
distributions, but also because it includes in the analysis the relation of 
commodity and stock markets. One could apply similar methodology in order to 
test the predictions of our model, in particular the relation between the 
correlation of the stock market and the commodity demand with the forward 
premia, the volume in forward contracts and the optimal storage amount. We leave 
this interesting task as a subject for future research. 
\end{remark}

\subsection{The effect of an existing hedge from a previous cycle}

Aside from uneven production levels, another factor that decreases the 
producers' tendency for storing is an already undertaken hedging position from a 
previous production and trading cycle. In this respect, we consider the 
situation where a forward contract with maturity $T$ was already issued during 
the previous cycle, and assume that producers have already taken a long position 
on it. We then examine how this existing hedge affects the spot and forward 
equilibrium quantities. More precisely, the total position of the producers' 
takes the following form (compare to \eqref{eq:producers_position})
\begin{align}\label{eq:producers_position_dh}
\underline{w}(\alpha,h^p)
 &= P_0(\pi_0-\alpha)(1+R)+P_T(\pi_T+\alpha(1-\varepsilon))+h^p(P_T-F) + h'(P_T-F'),
\end{align}
where $h'$ denotes the position in the forward contract with maturity at time 
$T$ that was bought with strike price $F'$ at the previous cycle. The 
optimization problem for the producers has again the same form as in 
\eqref{eq:producers_problem}, and in particular for the model of subsection 
\ref{subsec:bm model}, it takes the form of the following quadratic programming 
problem:
\begin{equation}\label{eq:producers_maximization_brown_dh}
\Pi^p
 = \underset{h^p\in\R,\alpha\in[0,\pi_0]}{\max}
   \big\{ \ud_1\alpha^2 + \ud_2'\alpha + \ud_3\alpha h^p
            + \ud_4(h^p)^2 + \ud_5'h^p+\ud_6\big\},
\end{equation}
where $\ud_1, \ud_3, \ud_4$ remain the same and are provided by 
\eqref{eq:constants_prod_prob}, while  
\begin{align*}
\ud_2' &= \frac{2(1+R)\pi_0-(1-\varepsilon)(2\pi_T+h')-(R+\varepsilon)\mu}{m}
            - \gamma_p(1-\varepsilon)(\pi_T+h') \Var[P_T] \nonumber\\
\ud_5' &= -\Big(F-\frac{\mu-\pi_T}{m}\Big) - \gamma_p(\pi_T+h') \Var[P_T].
\end{align*}

We observe that an existing long position in the forward contract, i.e.~$h'>0$, 
has the same impact as an increase of future production $\pi_T$. Therefore, as 
in Figures \ref{fig:BM-4}, positive $h'$ implies less storage and hence lower 
spot equilibrium price. In addition, as in Figure \ref{fig:BM-6}, when producers 
have already hedged some of their risk, the convenience yield increases, a fact 
that reflects better inventory management. Similarly and as expected, a gradual 
hedging effectively serves to stabilize commodity prices (see Figure 
\ref{fig:BM-5}). This is because the tendency to increase spot prices by forward 
contracts is gradually applied to prices. Intuitively, when producers hedge the 
future price uncertainty by rolling their position in the forward contract, the 
effect on the spot prices is spread through time. Note, however, that an 
intimate analysis on the gradually optimal hedging requires a dynamic version of 
our equilibrium model, which is left open for future research.

\bigskip

\appendix
\section{Proofs of Section \ref{sec:illu}}
\label{app-proofs}

\begin{proof}[Proof of Proposition \ref{prop:BM}]
The model with dynamics \eqref{eq:model-BM} fits in the framework of Section
\ref{model-ctdm} by considering a 2-dimensional Brownian motion $Z$ whose
characteristic triplet has the form
\begin{equation}\label{eq:example1}
b = \left(%
\begin{array}{c}
  b_1 \\ 0
\end{array}%
\right), \quad
c = \left(%
\begin{array}{cc}
  \sigma_1^2 & \rho \sigma_1 \sigma_2 \\
  \rho \sigma_1 \sigma_2 & \sigma_2^2
\end{array}%
\right)\quad \text{ and } \quad
\nu \equiv 0,
\end{equation}
where $b\in\R$, $\sigma_1,\sigma_2\in\R_+$ and $\rho\in[-1,1]$, while the
vectors $u_1,u_2\in\R^2$ have the form
\begin{equation}\label{eq:example2}
u_1 = \left(%
\begin{array}{c}
  1 \\ 0 \\
\end{array}%
\right)\quad \text{ and } \quad
u_2 = \left(%
\begin{array}{c}
  0 \\ 1 \\
\end{array}%
\right).
\end{equation}
The cumulant generating functions of $Y_1=\scal{u_1}{Z_1}$ and
$X_1=\scal{u_2}{Z_1}$ are given by
\begin{align}
\kappa_1(v) = vb_1 + \frac{v^2\sigma_1^2}{2}
 \quad \text{ and } \quad
\kappa_2(v) = \frac{v^2\sigma_2^2}{2}.
\end{align}
The set $\UUU_Z$ equals $\R^2$, thus Assumption \ref{cond:EM} and
\eqref{cond:boundedF} are trivially satisfied, while the same is true for
Assumption \ref{cond:FE} since $\nu\equiv0$. Assumption \ref{cond:COE} is also
satisfied due to $\kappa_2$ being quadratic in $u$, while Assumption
\ref{cond:NALevy} is fulfilled since we can construct a martingale measure under
which $Z$ remains a Brownian motion.

Starting with the producers' side, the optimal hedging and storage positions are
determined by Proposition \ref{pro:producers_problem} and \eqref{eq:up-Levy},
leading to the following quadratic programming problem:
\begin{equation}\label{eq:producers_maximization_brown}
\Pi^p
 = \underset{h^p\in\R,\alpha\in[0,\pi_0]}{\max}
   \big\{ \ud_1\alpha^2 + \ud_2\alpha + \ud_3\alpha h^p
            + \ud_4(h^p)^2 + \ud_5h^p+\ud_6\big\},
\end{equation}
where
\begin{align}\label{eq:constants_prod_prob}
\ud_1 &= - \left( \frac{1+R+(1-\varepsilon)^2}{m}
            + \frac{\gamma_p(1-\varepsilon)^2}{2} \Var[P_T] \right) \nonumber\\
\ud_2 &= \frac{2(1+R)\pi_0-2(1-\varepsilon)\pi_T-(R+\varepsilon)\mu}{m}
            - \gamma_p(1-\varepsilon)\pi_T \Var[P_T] \nonumber\\
\ud_3 &= -\frac{1-\varepsilon}{m} - \gamma_p(1-\varepsilon)\Var[P_T] \\
\ud_4 &= -\frac{\gamma_p}{2} \Var[P_T] \nonumber\\\nonumber
\ud_5 &= -\Big(F-\frac{\mu-\pi_T}{m}\Big) - \gamma_p\pi_T \Var[P_T].
\end{align}
The first order conditions yield the following solutions
\begin{align}
h^{p,*} = -\frac{\alpha \ud_3 + \ud_5}{2\ud_4}
 \quad \text{ and } \quad
\alpha^* = \frac{\ud_3\ud_5-2\ud_2\ud_4}{4\ud_1\ud_4-\ud_3^2}.
\end{align}
Therefore, the optimal strategy $(\hat\alpha,\hat{h}^p)\in[0,\pi_0]\times\R$
for the producers' problem is provided by
\begin{equation}\label{eq:producers_optimal_brown}
\hat\alpha = (\alpha^*\vee 0)\wedge\pi_0
 \quad \text{ and } \quad
\hat{h}^p = -\frac{\hat{\alpha}\ud_3 + \ud_5}{2\ud_4}.
\end{equation}

Next, we turn our attention to the investors' problem and follow the strategy 
outlined in subsection \ref{subs:spec_revisited}. The cumulant generating 
function of
$Z$ under $\P_\sss$ is provided by
\begin{align}
\kappa^\sss(v) = \scal{v-\xi}{b} + \frac{\scal{v-\xi}{c(v-\xi)}}{2},
\end{align}
where $\xi=-\frac{\gamma_s h^s}{m}u_2$. In particular, the characteristics of
$Y$ under $\P_\sss$ are
\begin{align}
b_1^\sss = b_1 - \rho\sigma_1\sigma_2 \frac{\gamma_s h^s}{m}
 \quad \text{ and } \quad
c_1^\sss = \sigma_1^2,
\end{align}
thus the characteristics of the exponential transform $\widetilde{Y}$ under
$\P_\sss$ are
\begin{align}
\widetilde{b}_1^\sss
 = b_1 - \rho\sigma_1\sigma_2 \frac{\gamma_s h^s}{m}
 + \frac{\sigma_1^2}{2}
 \quad \text{ and } \quad
\widetilde{c}_1^\sss = \sigma_1^2.
\end{align}
The cumulant generating function of $\widetilde{Y}$ simply has the form
\begin{align}
\widetilde{\kappa}_1^\sss (v)
 = v\Big(b_1 - \rho\sigma_1\sigma_2 \frac{\gamma_s h^s}{m} +
    \frac{\sigma_1^2}{2} \Big)
 + \frac{v^2\sigma_1^2}{2},
\end{align}
and its derivative obviously equals
\begin{align}
\frac{\partial }{\partial v} \widetilde{\kappa}_1^\sss (v)
 = b_1 - \rho\sigma_1\sigma_2 \frac{\gamma_s h^s}{m}
 + \frac{\sigma_1^2}{2} + v\sigma_1^2.
\end{align}
Therefore, the solution to equation \eqref{def:theta*} is
\begin{align}
\eta^* = \rho\frac{\sigma_2}{\sigma_1} \frac{\gamma_s h^s}{m}
       - \frac{b_1}{\sigma_1^2} - \frac12,
\end{align}
and the minimal entropy equals
\begin{align}
\mathcal{H}(\P_*|\P_\sss)
 = \frac{T}2 (\eta^*\sigma_1)^2
 = \frac{T}2 \Big( \lambda - \rho\sigma_2\frac{\gamma_s h^s}{m} \Big)^2.
\end{align}
Here $\lambda$ denotes the `market price of risk' for the asset $S$, i.e.
$\lambda=\frac{\mu_1-r}{\sigma_1}$, with $\mu_1$ being the expected rate of
return of $S$ and $r$ the continuously compounded interest rate, while we have
also used that $b_1=\mu_1-r-\sigma_1^2/2$.

The investors' optimal position in the forward contract is determined by
\eqref{eq:us-Levy}, leading to the following quadratic optimization problem:
\begin{align}\label{eq:investors_maximization_brown}
\Pi^s = \max_{h^s\in\R} \big\{ \ud_7(h^s)^2  + \ud_8h^s + \ud_9\big\},
\end{align}
where
\begin{align}
\ud_7 &= -\frac{\gamma_s}{2} (1-\rho^2) \Var[P_T]\\
\ud_8 &= \E[P_T] - F - \lambda\rho\sqrt{T}\sqrt{\Var[P_T]}.
\end{align}
Applying the first order conditions once again, we arrive at the optimal
position for the investors
\begin{align}
\hat{h}^s = \frac{\E[P_T]-F}{\bar{\gamma}_s\Var[P_T]}
      - \frac{\lambda\rho\sqrt{T}}{\bar{\gamma}_s\sqrt{\Var[P_T]}},
\end{align}
where $\bar{\gamma}_s=\gamma_s(1-\rho^2)$.
\end{proof}

\medskip

\begin{proof}[Proof of Proposition \ref{prop:LJD}]
The model with dynamics \eqref{ljd-X} fits in the framework of Section
\ref{model-ctdm} by considering a 2-dimensional \lev process $Z$ whose
characteristic triplet has the form
\begin{equation}\label{ljd-triplet-X}
b=\left(%
\begin{array}{c}
  b_1 \\
  b_2
\end{array}%
\right),\,\,\,\, c=\left(%
\begin{array}{cc}
  \sigma_1^2 & \rho\sigma_1\sigma_2 \\
  \rho\sigma_1\sigma_2 & \sigma_2^2
\end{array}%
\right)\,\,\,\,\text{ and }\,\,\,\,
\nu(\dx_1,\dx_2)= \lambda 1_{\{\eta_1,\eta_2\}}(\dx_1,\dx_2),
\end{equation}
while the vectors $u_1,u_2\in\R^2$ are provided by \eqref{eq:example2}. The
cumulant generating function of $Y_1=\scal{u_1}{Z_1}$ and
$X_1=\scal{u_2}{Z_1}$ is given by
\begin{align}\label{ljd-cumulant-12}
\kappa_i(v)
 &= vb_i + \frac{v^2\sigma_i^2}{2} + \lambda\left(\e^{v\eta_i}-1\right)
 \quad i=1,2.
\end{align}
The set $\UUU_Z$ equals $\R^2$, thus Assumption \ref{cond:EM} and
\eqref{cond:boundedF} are trivially satisfied, while the same is true for
Assumption \ref{cond:FE}. Assumption \ref{cond:COE} is also satisfied due to
$\kappa_2$ being quadratic in $v$, while Assumption \ref{cond:NALevy} is
satisfied since we can construct a martingale measure under which $X$ has
finite first moment.

Starting with the producers side, the optimal hedging and storage positions are
provided by Proposition \ref{pro:producers_problem} and \eqref{eq:up-Levy},
leading to the following optimization problem:
\begin{equation}\label{ljd-producers_maximization_brown}
\Pi^p
 = \underset{h^p\in\R,\alpha\in[0,\pi_0]}{\max} f(\alpha,h^p)
\end{equation}
where
\begin{align}
f(\alpha,h^p)
 := \ud_1\alpha^2 + \ud_2\alpha + \ud_3\alpha h^p + \ud_4(h^p)^2
  + \ud_5h^p + \ud_6 + j(\alpha,h^p),
\end{align}
with
\begin{align}
j(\alpha,h^p)
 := - \lambda \left( \e^{-\gamma_p\eta_2\ell(\alpha,h^p)} -1
      \right)\frac{T}{\gamma_p}.
\end{align}
The coefficients $\ud_1,\dots,\ud_5$ are provided by
\eqref{eq:constants_prod_prob} by replacing $\Var[P_T]$ with
$\frac{\sigma_2^2T}{m}$. The first order optimality conditions lead to the
system of non-linear equations \eqref{eq:jump_system_01}, i.e.
\begin{align}\label{eq:jump_system_1}
\begin{cases}
\frac{\partial }{\partial \alpha} f(\alpha,h^p)
 &= 2\ud_1\alpha + \ud_2 + \ud_3h^p +
    \frac{\lambda\eta_2T(1-\varepsilon)}{m}\e^{-\gamma_p\eta_2\ell(\alpha,h^p)}
  =0,\\
\frac{\partial }{\partial h^p} f(\alpha,h^p)
 &= \ud_3\alpha + 2\ud_4h^p + \ud_5 +
    \frac{\lambda\eta_2T}{m}\e^{-\gamma_p\eta_2\ell(\alpha,h^p)}
  =0,
\end{cases}
\end{align}
and its solution is denoted by $(\alpha^*,h^{p,*})$, where the relation of
$\alpha^*$ and $h^{p,*}$ is given by the following linear equation
\begin{align*}
\alpha^*
  = \frac{2(1-\varepsilon)\ud_4-\ud_3}{2\ud_1-(1-\varepsilon)\ud_3}h^{p,*}
  + \frac{(1-\varepsilon)\ud_5-\ud_2}{2\ud_1-(1-\varepsilon)\ud_3}.
\end{align*}
Therefore, the optimal strategy $(\hat\alpha,\hat{h}^p)\in[0,\pi_0]\times\R$
for the producers problem is provided by
\begin{equation}\label{ljd-producers_optimal_brown}
\hat\alpha = (\alpha^*\vee 0)\wedge\pi_0
 \quad \text{ and } \quad
\hat{h}^p = h^{p,*}(\hat\alpha).
\end{equation}

Next, we turn our attention to the investors problem and follow again the 
strategy of subsection \ref{subs:spec_revisited}. The characteristics of $Y$ 
under $\P_\sss$ are provided by Lemma \ref{lem:triplet}, thus using 
\eqref{ljd-triplet-X} we get that
\begin{align}
b_1^\sss &= b_1 - \rho\sigma_1\sigma_2\zeta
      + \lambda\eta_1(\e^{-\eta_2\zeta}-1) \nonumber\\
c_1^\sss &= \sigma_1^2 \\\nonumber
1_E(y)*\nu_1^\sss &= 1_E(\scal{u_1}{x}) \e^{\scal{\xi}{x}} * \nu,
\end{align}
where $\zeta:=\frac{\gamma_sh^s}{m}$, $E\in\mathcal{B}(\R)$ and ``$*$'' denotes
integration. Therefore, the cumulant generating function of $Y$ under $\P_\sss$
takes the form
\begin{align}
\kappa_1^\sss (v)
  = vb_1^\sss + \frac{v^2\sigma_1^2}{2} + \lambda (\e^{v\eta_1}-1)
    \e^{-\zeta\eta_2}.
\end{align}
Moreover, the characteristics of the exponential transform $\widetilde{Y}$ of
$Y$ are provided by \eqref{eq:exp-trans-triplet}, thus we obtain that
\begin{align}
\widetilde{b}^\sss_1 &= \kappa_1^\sss(1) \nonumber\\
\widetilde{c}_1^\sss &= \sigma_1^2 \\\nonumber
1_E(z)*\widetilde{\nu}_1^\sss &= 1_E(\e^y-1)*\nu_1^\sss.
\end{align}
Hence, the cumulant generating function of $\widetilde{Y}$ under $\P_\sss$
equals
\begin{align}
\widetilde{\kappa}^\sss_1 (v)
 &= v\widetilde{b}^\sss_1 + \frac{v^2\sigma_1^2}{2}
  + \lambda \e^{-\eta_2\zeta}
    \left(\e^{v(\e^{\eta_1}-1)}-1\right)
\end{align}
and its derivative with respect to $v$ equals
\begin{align}
\frac{\partial}{\partial v} \widetilde{\kappa}^\sss_1 (v)
 &= \widetilde{b}^\sss_1 + v\sigma_1^2
  + \lambda \e^{-\eta_2\zeta}\e^{v(\e^{\eta_1}-1)}(\e^{\eta_1}-1).
\end{align}
The minimal entropy martingale measure is determined by the solution $\eta_*$
to the non-linear equation
\begin{align}\label{eq:jump_system_2}
\widetilde{b}^\sss_1 + \eta_* \sigma_1^2
+ \lambda \e^{-\eta_2\zeta}\e^{\eta_*(\e^{\eta_1}-1)}(\e^{\eta_1}-1) = 0,
\end{align}
and then the minimal entropy equals
\begin{align}
\mathcal{H}(\P_*|\P_\sss)
  = -\frac{T}{\gamma_s} \widetilde{\kappa}^\sss_1 (\eta_*).
\end{align}
Now, putting the pieces together, the investors optimization problem takes the
form
\begin{align}
\Pi^s = \max_{h^s} g(h^s)
\end{align}
where
\begin{align}\label{eq:jump_system_3}
g(h^s)
 := -\frac{T}{\gamma_s} \left\{ \widetilde{\kappa}^\sss_1(\eta_*)
  + \kappa_2\left(-\frac{\gamma_sh^s}{m}\right) \right\}
  + h^s\big( \EE[P_T] -F \big),
\end{align}
and the maximizer is determined by the first order conditions, leading to
\eqref{eq:jump_system_02}.
\end{proof}

\bibliographystyle{plainnat}
\bibliography{references}

\end{document}

%% file: head.tex
\usepackage{times}
\usepackage{eurosym}
\usepackage{textcomp}
\usepackage{mathrsfs}
\usepackage{graphicx}
\usepackage{enumitem}
\usepackage{stmaryrd}
\usepackage{epsfig}
\usepackage[normalem]{ulem}

\newcommand{\R}{\mathbb{R}}

\usepackage[hyperref,amsmath,thmmarks]{ntheorem}

%% file: ahead.tex
\providecommand{\R}{} \renewcommand{\R}{{\mathbb R}}

\newcommand{\PP}{{\mathbb P}}
\newcommand{\QQ}{{\mathbb Q}}

\newcommand{\EE}{{\mathbb E}}

\newcommand{\HH}{{\mathcal H}}
\newcommand{\UU}{\mathbb{U}}

\newcommand{\MM}{{\mathcal M}}

\newcommand{\EN}{{\mathcal E}}

\newcommand{\NA}{\mathcal{N\!A}}

\newcommand{\Var}{\operatorname{\mathbb{V}ar}}

\newcommand{\scal}[2]{\ensuremath{\left\langle #1, #2 \right\rangle}}

\newcommand{\sbrac}[1]{\ensuremath{\left[#1\right]}}
\newcommand{\cbrac}[1]{\ensuremath{\left\{#1\right\}}}

\def\E{\mathbb{E}}

\def\Q{\mathbb{Q}}
\def\P{\mathbb{P}}
\def\R{\mathbb{R}}

\def\F{\mathcal{F}}
\def\bF{\mathbf{F}}

\def\Rd{\mathbb{R}^d}

\def\ud{\mathrm{d}}

\def\dx{\mathrm{d}x}

\def\dsdx{(\mathrm{d}s,\mathrm{d}x)}

\def\1{1}
\def\e{\mathrm{e}}

\def\lev{L\'evy\xspace}
\def\lk{L\'evy--Khintchine\xspace}

\def\lito{L\'evy--It\^o\xspace}

\def\procs{processes\xspace}

\newcommand{\prozess}[1][L]{{\ensuremath{#1=(#1_t)_{t\in[0,T]}}}\xspace}

\numberwithin{equation}{section}